\documentclass{article}
\usepackage{fullpage}
\usepackage{hyperref}
\usepackage{amsmath, amssymb}
\usepackage{cleveref}
\usepackage{mathtools}
\usepackage{amsthm} 
\usepackage{enumerate}  
\usepackage{url} %
\usepackage{makeidx} 
\usepackage{multirow} 
\usepackage{caption}
\usepackage{subcaption}

\usepackage{tikz}
\usetikzlibrary{decorations.markings}
\pgfdeclarelayer{background}
\pgfdeclarelayer{nodelayer}
\pgfdeclarelayer{edgelayer}
\pgfsetlayers{background,edgelayer,nodelayer,main}

\tikzstyle{new}=[circle,  minimum width=4pt,inner sep=0pt, fill=black,draw=black]
\tikzstyle{none}=[circle,fill=white,draw=black]
\tikzstyle{n}=[shape=rectangle,minimum width=1pt,inner sep=0pt, fill=none,draw=none]
\tikzstyle{emph}=[circle,  minimum width=4pt,inner sep=0pt, fill=magenta,draw=magenta]

\tikzset{directed/.style={decoration={
  markings,
  mark=at position .6 with {\arrow{>}}},postaction={decorate}}}
\tikzset{directed2/.style={decoration={
  markings,
  mark=at position .6875 with {\arrow{>}},
  mark=at position .7125 with {\arrow{>}}},postaction={decorate}}}
\tikzset{directed3/.style={decoration={
  markings,
  mark=at position .7 with {\arrow{>}}},postaction={decorate}}}

\usepackage{array}
\usepackage{float}
\usepackage{wrapfig}

\usepackage{graphicx}
\usepackage{rotating}
\usepackage{nicefrac}
\usepackage{longtable}
\usepackage{xcolor}

\usepackage[overload]{textcase}

\makeatletter
\def\cleardoublepage{
  \clearpage
  \if@twoside\ifodd\c@page\else
  \hbox{}
  \thispagestyle{empty}
  \newpage
  \if@twocolumn\hbox{}\newpage\fi
  \fi\fi
  }
\makeatother

\newtheoremstyle{plainsl}%
	{\topsep}
	{\topsep}
	{\slshape} % only non-default setting
	{}
	{\normalfont\bfseries}
	{.}
	{ }
	{}

{\theoremstyle{plainsl}
\newtheorem{theorem}{Theorem}[section]
\newtheorem{lemma}[theorem]{Lemma}

\newtheorem{corollary}[theorem]{Corollary}

\newtheorem{openprob}[theorem]{Open Problem}

}
{\theoremstyle{remark}
\newtheorem{remark}[theorem]{Remark}

\newtheorem{definition}[theorem]{Definition}}

\numberwithin{equation}{section}

\newcommand\comp[1]{{\mkern2mu\overline{\mkern-2mu#1}}}
\newcommand\pmat[1]{\begin{pmatrix} #1 \end{pmatrix}}

\DeclareMathOperator\ev{ev}

\newcommand\cx{{\mathbb C}}% complexes

\newcommand\ints{{\mathbb Z}}
\newcommand\re{{\mathbb R}}%reals
\newcommand\rats{{\mathbb Q}}
\newcommand\Zv{{\mathbf v}}

\newcommand\ones{{\mathbf 1}}

\newcommand\ket[1]{| #1 \rangle}
\newcommand\braket[2]{\langle #1| #2 \rangle}

\definecolor{vcolour}{RGB}{230,97,0}
\definecolor{kcolour}{RGB}{93,58,155}
\definecolor{gcolour}{RGB}{40,155,50}

\usepackage{authblk}

\newcommand{\arxiv}[1]{\href{https://arxiv.org/abs/#1}{\texttt{arXiv:#1}}}

\title{Peak state transfer in continuous quantum walks}
\author[1]{Gabriel Coutinho}
\author[2,3]{Krystal Guo}
\author[2,3]{Vincent Schmeits}

\affil[1]{Federal University of Minas Gerais, Belo Horizonte, Brazil}
\affil[2]{Korteweg-de Vries Institute for Mathematics, University of Amsterdam, Amsterdam, The Netherlands}
\affil[3]{\href{https://qusoft.org/}{QuSoft} (Research center for quantum software \& technology), Amsterdam, The Netherlands}

\affil[ ]{\texttt{gabriel@dcc.ufmg.br, k.guo@uva.nl, v.f.schmeits@uva.nl}}

\date{Sep 9, 2026}
\begin{document}

\maketitle
\begin{abstract}

We introduce and study \emph{peak state transfer}, a  notion of high  state transfer in qubit networks modeled by continuous-time quantum walks. Unlike perfect or pretty good state transfer, peak state transfer does not require fidelity arbitrarily close to 1, but crucially allows for an explicit determination of the time at which transfer occurs. We provide a spectral characterization of peak state transfer, which allows us to find many examples of peak state transfer, and we also establish tight lower bounds on fidelity and success probability. As a central example, we construct a family of weighted path graphs that admit peak state transfer over arbitrarily long distances with transfer probability approaching $\pi/4 \approx 0.78$, and we show how these are directly related to a family of unweighted graph which can also achieve peak state transfer after minor adjustments. Our analysis includes an explicit spectral characterization of these graphs which leads to the conclusion that these constructions offer improved robustness compared to classical perfect state transfer examples, such as the weighted paths derived from hypercubes, thereby identifying them as practical candidates for cheap and efficient quantum wires.

\vspace{10pt}

 
\end{abstract}

% 81P45 – Quantum information, communication, networks

% 05C50 – Graphs and linear algebra (spectral graph theory)

% 05C90 – Applications of graph theory

% 81Q99 – Quantum theory (none of the above, but in this section)

\section{Introduction}\label{sec:intro}

The ability to transmit quantum information with high fidelity is a fundamental requirement for quantum computing.
One approach to achieving this is via a \textit{quantum wire}:  a network of interacting qubits that allows for a quantum state input at one of the qubits to be transferred to another site by means of the time-evolution of a Hamiltonian. This problem is not new; it has been studied in multiple works since the seminal paper \cite{BoseQuantumComPaths}, and several connections to the combinatorics of the network (a graph, from now on) have been found \cite{ChrDatEke2004, ChristandlPSTQuantumSpinNet2, GodsilStateTransfer12,KayReviewPST, CoutinhoGodsilGuoVanhove2, KendonTamon, CoutinhoGodsilPSTProducts, KayBeyond, KemLipTau2017, GodGuoKem2020}. This paper is motivated by the very natural goal of minimizing the size of the graph --- with respect to both the number of qubits (vertices of the graph) and the number of interacting terms (edges of the graph) --- while maintaining a large distance between the input and output qubits. In other words, we aim to build a cheap and long quantum wire.

Ideally, our quantum wire should transfer the input qubit state with probability 1 to the target qubit, and this indeed occurs in several known examples of graphs \cite{CoutinhoPhD,CoutinhoGodsilGuoVanhove2}. This perfect state transfer is, however, not known to be available at arbitrary distances unless the size of the graph grows exponentially with the distance of the state transfer (see \cite{ChrDatEke2004}), which is undesirable for practical applications. The most recent attempt to construct a family of graphs admitting perfect state transfer and whose size is bounded by a polynomial in the graph diameter was \cite{kay2018perfect}, where the only gain achieved was in the base of the exponent.

Since it is understood that achieving \textsl{perfect} state transfer is rare, we move to the relaxed notion of \textsl{pretty good} (or \textsl{almost perfect}) state transfer, which has been considered in \cite{GodsilKirklandSeveriniSmithPGST,VinetZhedanovAlmost}. This phenomenon requires the network to provide arbitrarily good fidelity of state transfer at increasingly large times. There are known examples of graphs whose distance between the input and output qubits grow while maintaining the size of the graph linearly bounded on this distance \cite{CouGuovBo2017}, which would give rise to a cheap quantum wire, except that all known characterizations of this phenomenon rely heavily on number-theoretic tools and  do not provide a way of computing the times at which high state transfer occurs; see for example \cite{coutinho2023irrational}. In other words, we can prove that pretty good state transfer occurs but we cannot predict when it occurs, which presents a significant impediment for a cheap quantum wire. We illustrate an example of pretty good state transfer in Figure \ref{fig:pgst-p9}.

\begin{figure}[htbp]
    \centering
      \begin{minipage}{0.35\textwidth}
    \centering
\begin{tikzpicture}[scale=0.6, every node/.style={circle, draw, minimum size=2mm, inner sep=1pt}]
  % Draw the nodes
  \foreach \i in {0,...,8} {
    \node (\i) at (\i, 0) {\i};
  }

  % Draw the edges
  \foreach \i in {0,...,7} {
    \pgfmathtruncatemacro{\next}{\i + 1}
    \draw (\i) -- (\next);
  }
\end{tikzpicture}
    \par\vspace{1ex}
    \textbf{(a)} the path graph on $9$ vertices
  \end{minipage}
  \quad
  \begin{minipage}{0.6\textwidth}
    \centering
    \includegraphics[width=\linewidth]{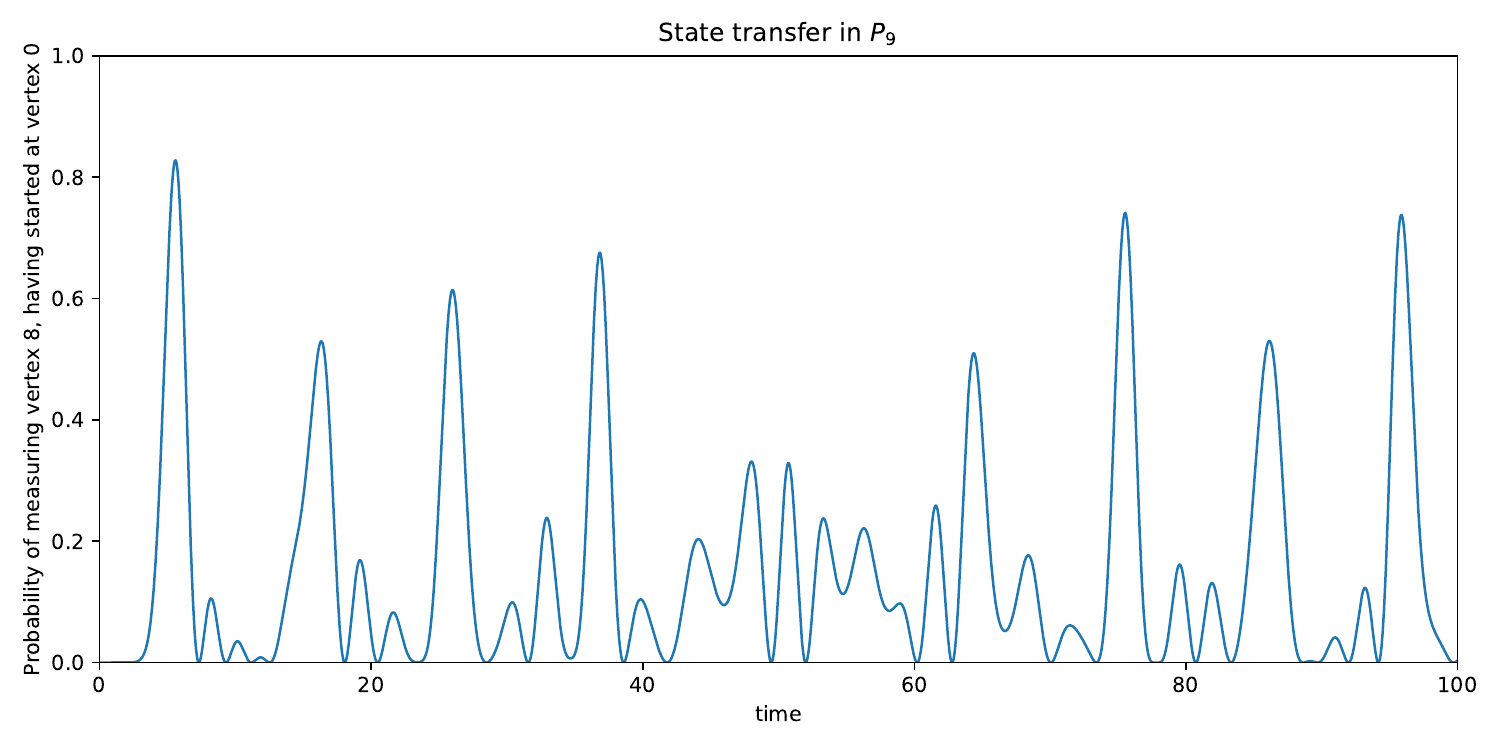}
        \par\vspace{1ex}
        \textbf{(b)} $(U(t)\circ \comp{U(t)})_{0,8}$
  \end{minipage}
    \caption{The probability of measuring at vertex 8, having started at vertex 0, from time $0$ to $100$, in the path on $9$ vertices. This graph has pretty good state transfer between vertices $0$ and $8$, thus the curve will come arbitrarily close to the $y=1$ line. However, as we see, it is not approaching $1$ very quickly and we cannot easily find, for example, a time when the state probability exceeds 0.9. }
    \label{fig:pgst-p9}
\end{figure}
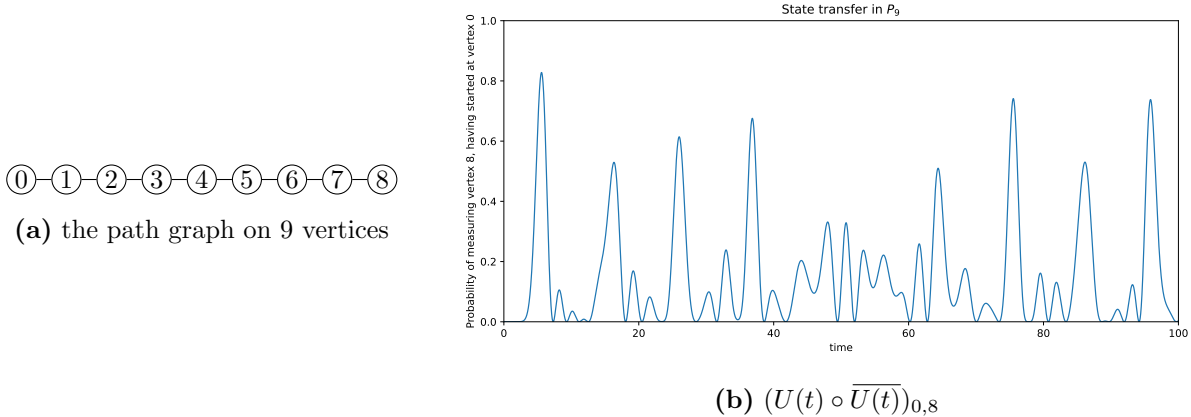

We have finally arrived at this present work, wherein we give up the guarantee that the state transfer is arbitrarily close to perfect, in exchange for a mathematical framework in which the time of state transfer can be obtained. Here we introduce the relaxed notion of \textsl{peak state transfer}, study its properties, and present a certain family of graphs which achieves this phenomenon at arbitrarily long distances while keeping the size of the graphs bounded by a polynomial on the distance. In peak state transfer the probability of success is possibly bounded away from $1$, but it is still controlled.  For discrete-time quantum walks, two of the authors defined peak state transfer, using the projected transition matrix, in \cite{GuoSch2025}. 

We give a spectral characterization of peak state transfer in Theorem \ref{thm:peakST-char-time}. This allows us to determine if peak state transfer occurs, given the spectral information of a graph, as well as the time at which it occurs. 

Operationally, peak state transfer provides a protocol with a predetermined readout time. Once the spectral data determine a peak time $\tau$, the input state is prepared at the source vertex, the system evolves until exactly time $\tau$, and a measurement is performed to determine whether the transfer has succeeded. Unlike perfect state transfer, the success probability need not equal $1$: for the transfer of a single excitation it is $B(M)_{v,u}^{2}$, while for an arbitrary input qubit $\ket{\psi}=\alpha\ket{0}+\beta\ket{1}$, the postselection described in \Cref{sec:fidelity} succeeds with probability at least $B(M)_{v,u}^{2}$. Conditioned on success and after correcting the known phase, Lemma~\ref{lem:fidelity} gives an explicit lower bound on the fidelity of the output state. Thus, peak state transfer may not admit deterministic success, but still keeps an exactly known transfer time together with quantitative guarantees on both success probability and fidelity in a much broader family of examples, as we will see.

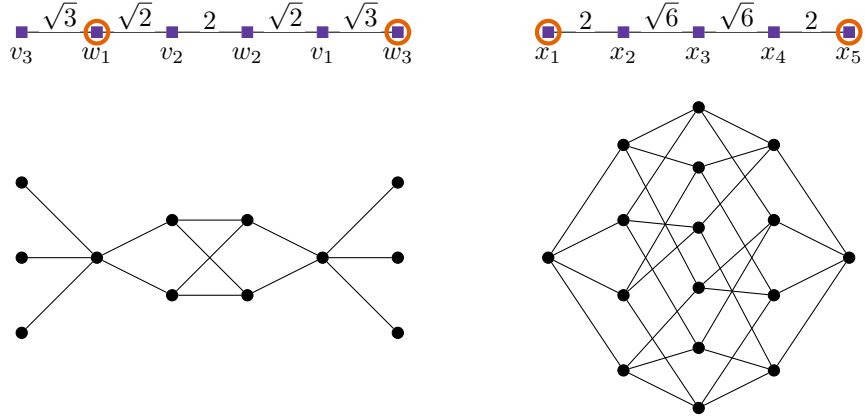
\begin{figure}
    \centering
\begin{tikzpicture}[scale=1.0, every node/.style={scale=1.0}]
% Define styles
\tikzstyle{graph node}=[circle, fill=black, draw=black, inner sep=1.5pt]
\tikzstyle{path node v}=[rectangle, fill=kcolour, draw=kcolour, inner sep=2pt]
\tikzstyle{path node w}=[rectangle, fill=kcolour, draw=kcolour, inner sep=2pt]
\tikzstyle{edge label}=[midway, fill=white, inner sep=1pt]

% ----- First path -----
\node[path node v] (x1) at (0, -1.5) {};
\node[path node w] (x2) at (1, -1.5) {};
\node[path node v] (x3) at (2, -1.5) {};
\node[path node w] (x4) at (3, -1.5) {};
\node[path node v] (x5) at (4, -1.5) {};
\node[path node w] (x6) at (5, -1.5) {};

\draw (x1) -- node[edge label, above] {$\sqrt{3}$} (x2);
\draw (x2) -- node[edge label, above] {$\sqrt{2}$} (x3);
\draw (x3) -- node[edge label, above] {$2$} (x4);
\draw (x4) -- node[edge label, above] {$\sqrt{2}$} (x5);
\draw (x5) -- node[edge label, above] {$\sqrt{3}$} (x6);

\draw[ultra thick, vcolour] (x2) circle [radius=1.5mm];
\draw[ultra thick, vcolour] (x6) circle [radius=1.5mm];

\node[below] at (x1.south) {$v_3$};
\node[below] at (x2.south) {$w_1$};
\node[below] at (x3.south) {$v_2$};
\node[below] at (x4.south) {$w_2$};
\node[below] at (x5.south) {$v_1$};
\node[below] at (x6.south) {$w_3$};

% ----- Second path (offset to the right) -----
\node[path node v] (y1) at (7, -1.5) {};
\node[path node w] (y2) at (8, -1.5) {};
\node[path node v] (y3) at (9, -1.5) {};
\node[path node w] (y4) at (10, -1.5) {};
\node[path node v] (y5) at (11, -1.5) {};

\draw (y1) -- node[edge label, above] {$2$} (y2);
\draw (y2) -- node[edge label, above] {$\sqrt{6}$} (y3);
\draw (y3) -- node[edge label, above] {$\sqrt{6}$} (y4);
\draw (y4) -- node[edge label, above] {$2$} (y5);
\draw[ultra thick, vcolour] (y1) circle [radius=1.5mm];
\draw[ultra thick, vcolour] (y5) circle [radius=1.5mm];
\node[below] at (y1.south) {$x_1$};
\node[below] at (y2.south) {$x_2$};
\node[below] at (y3.south) {$x_3$};
\node[below] at (y4.south) {$x_4$};
\node[below] at (y5.south) {$x_5$};
\begin{scope}[shift={(0, -3)}]
  % Everything in here is moved down by 2.5 units

\node[graph node] (x1) at (0, -1.5) {};
\node[graph node] (x1a) at (0, -2.5) {};
\node[graph node] (x1b) at (0, -0.5) {};
\node[graph node] (x2) at (1, -1.5) {};
\node[graph node] (x3) at (2, -1) {};
\node[graph node] (x3a) at (2, -2) {};
\node[graph node] (x4) at (3, -1) {};
\node[graph node] (x4a) at (3, -2) {};
\node[graph node] (x5) at (4, -1.5) {};
\node[graph node] (x6) at (5, -1.5) {};
\node[graph node] (x6a) at (5, -0.5) {};
\node[graph node] (x6b) at (5, -2.5) {};

\draw (x1) -- (x2);
\draw (x1a) -- (x2);
\draw (x1b) -- (x2);
\draw (x2) --(x3);
\draw (x2) --(x3a);
\draw (x3) --  (x4);
\draw (x3a) --  (x4);
\draw (x3) --  (x4a);
\draw (x3a) --  (x4a);
\draw (x4) --  (x5);
\draw (x4a) --  (x5);
\draw (x5) --  (x6);
\draw (x5) --  (x6a);
\draw (x5) --  (x6b);

% ----- Second path (offset to the right) -----
\node[graph node] (y1111) at (7, -1.5) {};
\node[graph node] (y1110) at (8, -3) {};
\node[graph node] (y1101) at (8, -2) {};
\node[graph node] (y1011) at (8, -1) {};
\node[graph node] (y0111) at (8, 0) {};

\node[graph node] (y1100) at (9, -3.5) {};
\node[graph node] (y1010) at (9, -2.7) {};
\node[graph node] (y0110) at (9, -1.9) {};
\node[graph node] (y1001) at (9, -1.1) {};
\node[graph node] (y0101) at (9, -0.3) {};
\node[graph node] (y0011) at (9, 0.5) {};

\node[graph node] (y1000) at (10, -3) {};
\node[graph node] (y0100) at (10, -2){};
\node[graph node] (y0010) at (10, -1){};
\node[graph node] (y0001) at (10,  0 ) {};
\node[graph node] (y0000) at (11, -1.5) {};

\draw (y1111) --  (y1110);
\draw (y1111) --  (y1101);
\draw (y1111) --  (y1011);
\draw (y1111) --  (y0111);

\draw (y1101) --  (y1100);
\draw (y1101) --  (y1001);
\draw (y1101) --  (y0101);

\draw (y1110) --  (y1100);
\draw (y1110) --  (y1010);
\draw (y1110) --  (y0110);

\draw (y1011) --  (y1010);
\draw (y1011) --  (y1001);
\draw (y1011) --  (y0011);

\draw (y0111) --  (y0110);
\draw (y0111) --  (y0101);
\draw (y0111) --  (y0011);

\draw (y0010) --  (y0011);
\draw (y0010) --  (y0110);
\draw (y0010) --  (y1010);

\draw (y0001) --  (y0011);
\draw (y0001) --  (y0101);
\draw (y0001) --  (y1001);

\draw (y0100) --  (y0101);
\draw (y0100) --  (y0110);
\draw (y0100) --  (y1100);

\draw (y1000) --  (y1001);
\draw (y1000) --  (y1010);
\draw (y1000) --  (y1100);

\draw (y0001) -- (y0000);
\draw (y0010) -- (y0000);
\draw (y0100) -- (y0000);
\draw (y1000) -- (y0000);

\end{scope}
\end{tikzpicture}

    \caption{\textit{Left:} an example of a weighted graph with peak state transfer between vertices at distance $4$ from Section \ref{sec:fam-high-peak-wtd-path}, with its realization as an unweighted graph under it. \textit{Right:} the weighted path coming from the hypercube of dimension $4$, with the hypercube under it.  }
    \label{fig:peak1}
\end{figure}

With peak state transfer, the output state need not be exactly the same as the input state. We give  a tight lower bound on the fidelity of the output state with the input state, as well as a lower bound on the probability of successfully transferring a qubit state with peak state transfer.

In Section \ref{sec:fam-high-peak-wtd-path}, we illustrate the usefulness of peak state transfer by giving a special infinite family admitting peak state state transfer where the transfer probability is approximately $\pi/4 \approx 0.78$. It is also a family of graphs for which transfer does not add a relative phase at $\ket 1$ (as opposed to many other known examples of perfect or pretty good state transfer), hence these wires are also good to transfer arbitrary qubit states.

This family of weighted paths is related to a family of unweighted graphs that admit peak state transfer in an equivalent manner; see Figure \ref{fig:peak1}. In these unweighted graphs, the peak state transfer occurs over distance $2d$ from $1$ vertex on one side to  $d+1$ vertices on the other side, the corresponding probabilities summing to approximately $0.78$ as well. With respect to the transfer distance, the number of vertices of these graphs is quadratic and the number of edges is cubic, resulting in a cheap wire.

The robustness comparison can be made quantitative by considering a readout at time $\tau+\delta$, rather than at the prescribed peak time $\tau$. Since the transfer probability attains its global maximum at $\tau$, its first-order variation vanishes, and the leading loss caused by a small timing error is quadratic in $\delta$. Standard sensitivity estimates bound this quadratic loss in terms of the spectral spread $\Delta(M)=\theta_{\max}-\theta_{\min}$; see Section~III.C of \cite{KayBeyond}, and also Theorem~2.9 of \cite{Kirkland2015} for an estimate on the derivatives of all orders. For our family transferring across distance $2d$, the spectral spread is $2(d+1)$, whereas the standard hypercube-derived perfect-transfer chain over the same distance has spectral spread $4d$. Consequently, the ratio between the corresponding quadratic sensitivity bounds is
\[
\left(\frac{2(d+1)}{4d}\right)^{2}
=\left(\frac{d+1}{2d}\right)^{2},
\]
which converges to $1/4$. Thus, within these spectral sensitivity estimates, our construction asymptotically reduces the leading effect of a fixed timing error by a factor of four compared with the hypercube-derived example.

Regarding the organization of this paper, we give the mathematical background for continuous-time quantum walks in Section \ref{sec:background}. We define peak state transfer rigorously in Section \ref{sec:peak}, look at the fidelity of transfer in Section \ref{sec:fidelity} and then give the spectral characterization in Section \ref{sec:main}. Examples of peak state transfer with respect to adjacency and Laplacian matrices are given in Section \ref{sec:exs}, with the most important family of examples in Section \ref{sec:fam-high-peak-wtd-path}. For the reader’s convenience, we include a table of nomenclature in Section \ref{app:nomen}; this serves as a quick reference for symbols used throughout the paper.

\section{Continuous-time quantum walks}\label{sec:background}

In this section, we will succinctly present the mathematical setup. We model a network of $n$ interacting qubits, where the interactions are defined by a time-independent Hamiltonian; that is, a symmetric matrix that acts on the Hilbert space of dimension $2^n$. The Pauli matrices are 
\[
\sigma_x = \pmat{0 & 1 \\ 1 & 0} , \quad \sigma_y = \pmat{0 & -i\\ i & 0} , \quad \sigma_z = \pmat{1 & 0 \\ 0 & -1} .
\]
We consider a graph $G$ on vertex set $V(G)$ and edge set $E(G)$, assuming $|V(G)| = n$, and define $n$-fold Kronecker products as follows:
\[
\sigma_u^x = I \otimes \cdots \otimes I\otimes  \sigma_x \otimes I \otimes \cdots \otimes I
\]
where $\sigma_x$ appears in the $u$th position, out of the $n$ factors of the Kronecker product. We define $\sigma_u^y, \sigma_u^z$ analogously. 

The \textsl{$XY$-Hamiltonian} is as follows:
\[
H_{XY} = \frac{1}{2} \sum_{uv \in E(G)} J_{uv} (\sigma_u^x\sigma_v^x + \sigma_u^y\sigma_v^y) . 
\]
The \textsl{$XYZ$-Hamiltonian} is as follows:
\[
H_{XYZ} = -\frac{1}{2} \sum_{uv \in E(X)} J_{uv} (\sigma_u^x\sigma_v^x + \sigma_u^y\sigma_v^y + \sigma_u^z\sigma_v^z - I_{2^n})  . 
\]
Taking $J_{uv} = 1$ for all $u,v$, it is straightforward to show that the action of $H_{XY}$ on the subspace spanned by the single excitation states $\ket{10\cdots0}$, $\ket{010\cdots 0}$, etc. is equivalent to the action of the adjacency matrix $A(G)$ on $\cx^n$, with the elementary basis.  Similarly, the action of $H_{XYZ}$ on the single excitation states is equivalent to the action of the Laplacian matrix $L(G)$ on $\cx^n$, with the elementary basis. We recall that the \textsl{adjacency matrix} of a graph $G$, denoted $A(G)$ (or simply $A$ when the context is clear), is given by 
\[
A(G)_{u,v} = \begin{cases}
    1, &\text{ if } u,v \text{ are adjacent}; \\ 
    0, &\text{otherwise.}
\end{cases}
\]
If the edges of the graph have weights, then we put those weights in the corresponding entries instead of a $1$. The \textsl{Laplacian matrix} of a graph $G$, denoted $L(G)$ (or simply $L$ when the context is clear), is given by 
\[
L(G)_{u,v} = \begin{cases}
    d(u), &\text{ if } u= v; \\
    -1, &\text{ if } u,v \text{ are adjacent};\\
    0, &\text{otherwise}
\end{cases}
\]
where $d(u)$ is the \textsl{degree} of vertex $u$, which is the number of edges incident with $u$. 
 
In this paper, we consider a \textsl{continuous-time quantum walk} with transition matrix $U(t)= e^{itM}$ where $M$ is an $n\times n$ symmetric matrix. 

\section{Peak state transfer}\label{sec:peak}

Let $M$ be a $n\times n$ symmetric matrix. The \textsl{spectral decomposition} of $M$ is the following:
\[
M = \sum_{r=0}^d \theta_r E_r,
\]
where
\begin{itemize}
    \item $\theta_0,\ldots, \theta_d$ are the distinct eigenvalues of $M$; and
    \item $E_r$ is the idempotent projection onto the $\theta_r$-eigenspace. 
\end{itemize}
By this we mean that $E_r$ is the unique matrix such that $E_r^2 = E_r$ and, for any $\Zv \in \re^n$, the image $E_r \Zv$ is an eigenvector of $M$ with eigenvalues $\theta_r$. 
Since $M$ is a symmetric matrix, the eigenvalues $\theta_0,\ldots, \theta_d$ are real numbers and $E_0,\ldots, E_d$ are real symmtric matrices. If $f(x)$ is any function which is defined on $\theta_0,\ldots, \theta_d$, then we may define
\[
f(M) = \sum_{r=0}^d f(\theta_r) E_r.
\]
In particular, the transition matrix of the continuous-time quantum walk on $M$ can be written as the following matrix-valued function in time:
\[
U(t) = e^{itM} = \sum_{r=0}^d e^{it\theta_r} E_r. 
\]
We will upper-bound the quantity $|U(t)_{v,u}|$, for all values of $t$, using this expression. 

A pivotal tool used in the study of perfect state transfer is a series of three inequalities where the conditions for equality were used in the characterization of perfect state transfer. These inequalities are described at length in \cite[Section 7.1]{CouGodBook}. Here we will only apply the first inequality; that is, we apply the triangle-inequality to see that 
\[
|U(t)_{v,u}| = \left| \sum_{r=0}^d e^{it\theta_r} (E_r)_{v,u}\right| \leq   \sum_{r=0}^d \left| e^{it\theta_r} (E_r)_{v,u}\right|  =   \sum_{r=0}^d \left| (E_r)_{v,u}\right|  =   \left( \sum_{r=0}^d \left| (E_r)\right| \right)_{v,u}
\]
since $ e^{it\theta_r}$ lies on the unit circle. We will refer to 
\[
B(M) \coloneq \sum_{r=0}^d \left| (E_r)\right|
\]
as the \textsl{bounding matrix} of $M$, as the $(v,u)$-entry of $B(M)$ upper-bounds  $|U(t)_{v,u}|$ for all values of $t$. For distinct vertices $u,v$, we say that there is  \textsl{peak state transfer} from $u$ to $v$ with respect to $M$ if there exists a time $\tau$ such that 
\[
|U(\tau)_{v,u}|  = B(M)_{v,u}. 
\]
Since $M$ is a real symmetric matrix, there is peak state transfer from $u$ to $v$ at time $\tau$ if and only if there is peak state transfer from $v$ to $u$ at time $\tau$; for this reason, we will sometimes refer to peak state transfer as being `between $u$ and $v$'.

\begin{remark}
 For vertices $u,v$, it is possible that $B(M)_{v,u} = 0$, in which case $U(t)_{v,u}=0$ for all $t$.  We call this \textsl{zero state transfer} between $u$ and $v$. If $M$ is the adjacency matrix of a graph $G$, this implies that all polynomials in $M$ have $0$ in the $u,v$ position and thus $u,v$ are in different connected components of $G$. However, there exist examples of zero state transfer in weighted graphs where the underlying graph is connected, albeit with some negative weights.

 Consider the cycle graph $C_4$ with one edge weighted with $-1$. See Figure \ref{fig:0st-transfer} for the graph, the weighted adjacency matrix and the bounding matrix. Since the $(0,2)$ and $(1,3)$ entries are $0$ in the bounding matrix, this weighted graph has zero state transfer from $0$ to $2$ and from $1$ to $3$.
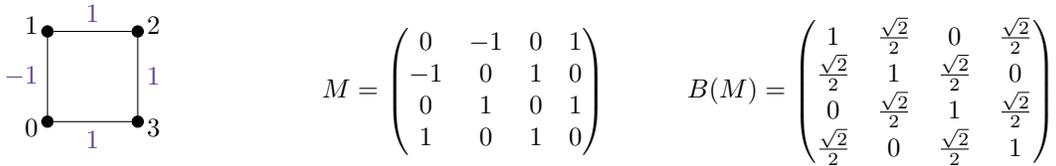
\begin{figure}[htbp]
  \begin{minipage}{0.3\textwidth}
  \centering
 \begin{tikzpicture}[scale=0.6, every node/.style={scale=1.0}]
\tikzstyle{graph node}=[circle, fill=black, draw=black, inner sep=1.5pt]
\node[graph node] (0) at (0,0) {};
\node[graph node] (1) at (2,0) {};
\node[graph node] (2) at (2,2) {};
\node[graph node] (3) at (0,2) {};
\node[left] at (0.south) {$0$};
\node[right] at (1.south) {$3$};
\node[right] at (2.north) {$2$};
\node[left] at (3.north) {$1$};
\draw (0) -- (1) -- (2) -- (3) -- (0);
\node[left] at (0,1) {\textcolor{kcolour}{$-1$}};
\node[right] at (2,1) {\textcolor{kcolour}{$1$}};
\node[below] at (1,0) {\textcolor{kcolour}{$1$}};
\node[above] at (1,2) {\textcolor{kcolour}{$1$}};
\end{tikzpicture}
  \end{minipage}
  \begin{minipage}{0.3\textwidth}
  \[
M =   \pmat{0 & -1 & 0 & 1 \\
-1 & 0 & 1 & 0 \\
0 & 1 & 0 & 1 \\
1 & 0 & 1 & 0
}
  \]
\end{minipage}
\quad
  \begin{minipage}{0.3\textwidth}
  \[
B(M)=   \pmat{1 & \frac{\sqrt{2}}{2} & 0 & \frac{\sqrt{2}}{2}\\
\frac{\sqrt{2}}{2}& 1 & \frac{\sqrt{2}}{2} & 0 \\
0 & \frac{\sqrt{2}}{2} & 1 & \frac{\sqrt{2}}{2}\\
\frac{\sqrt{2}}{2} & 0& \frac{\sqrt{2}}{2} & 1
}
  \]
\end{minipage}
\caption{\textit{From left to right:} a weighted graph on vertices $0,1,2,3$ with edge weights, the weighted adjacency matrix, and the corresponding bounding matrix.  \label{fig:0st-transfer}}
\end{figure}

\end{remark}

\section{Fidelity of peak state transfer}\label{sec:fidelity}

In the original setting of state transfer studied in \cite{BoseQuantumComPaths,ChrDatEke2004}, the aim is to move an arbitrary unknown $1$-qubit state $\ket{\psi} = \alpha\ket{0} + \beta\ket{1}$ between the ends of a `quantum wire', for instance as a subroutine for quantum algorithms, or to move `quantum information' between quantum computers. This is to happen by time evolution corresponding to a suitable $XY(Z)$-Hamiltonian $H$ on the $n$-qubit wire, in order to minimize the amount of necessary outside interactions with the wire. Let $\ket{0^n}$ denote the $n$-qubit ground state $\ket{0\ldots 0}$ and let $\ket{j} = \ket{0 \ldots 010 \ldots 0}$ denote the single excitation state with a $\ket{1}$ on the $j$-th qubit (the difference between the qubit $\ket{1} \in \cx^2$ and the single excitation state $\ket{1}\in (\cx^2)^{\otimes n}$ should be clear from the context). If the initial state $\ket{\Psi(0)}$ of the system is given by
\[
\ket{\Psi(0)} = \ket{\psi} \otimes \ket{0^{n-1}} = \alpha\ket{0^n}  + \beta\ket{1}
\]
the goal is to bring it to the state
\[
\ket{\Phi} = \ket{0^{n-1}} \otimes \ket{\psi} = \alpha\ket{0^n} + \beta\ket{n}.
\]
Since $\ket{0^{n}}$ is an eigenstate of $H$ for the eigenvalue $0$ and since $H$ acts invariantly on the subspace spanned by the single excitation states $\ket{j}$, the state $\ket{\Psi(t)}$ of the system at time $t$ can be written as
\[
\ket{\Psi(t)} = \alpha e^{itH}\ket{0^n} + \beta e^{itH}\ket{1} = \alpha\ket{0^n} + \beta \sum_{j=1}^n U(t)_{j,1}\ket{j},
\]
where $U(t) = e^{itM}$ for the adjacency or Laplacian matrix $M$ of the underlying graph describing the qubit interactions. If there is perfect state transfer from vertex $1$ to vertex $n$ at time $\tau$, then $U(\tau)_{n,1} = e^{i\theta}$ for some phase factor $e^{i\theta}$ and $U(\tau)_{j,1} = 0$ for $j \neq n$, hence $\ket{\Psi(\tau)} = \ket{0^{n-1}} \otimes (\alpha\ket{0} + \beta e^{i\theta}\ket{1})$. After correcting for the phase $\theta$, which does not depend on $\ket{\psi}$, the desired state $\ket{\Phi}$ has been attained.

Peak state transfer is a relaxation of perfect state transfer, hence it is a more prevalent property (see Table \ref{tab:peakST_table}), but it comes at a cost. There are two aspects to this cost: with peak state transfer, the probability of transferring the state is not necessarily $1$, and the transferred state is not exactly the same as the input state. We will make these statements precise in the following lemma. The \textsl{fidelity} between a pair of quantum states $\ket{\psi}$ and $\ket{\phi}$ is defined by $F(\ket{\psi},\ket{\phi}) \coloneqq |\braket{\psi}{\phi}|$. Here we give a lower bound; Bose \cite{BoseQuantumComPaths} instead looked at the average fidelity over all $1$-qubit states.

\begin{lemma}\label{lem:fidelity}
Let $H \in \{H_{XY},H_{XYZ}\}$ be such that there is peak state transfer from vertex $1$ to vertex $n$ at time $\tau$ in the underlying graph. Moreover, let the initial state of the system be
\[
\ket{\psi} \otimes \ket{0^{n-1}}
\]
for some arbitrary qubit state $\ket{\psi} = \alpha\ket{0} + \beta\ket{1}$. Measuring the first $n-1$ qubits yields the state $\ket{0^{n-1}} \otimes \ket{\phi}$ with probability
\[
p \coloneqq |\alpha|^2 + |\beta|^2 B(M)_{n,1}^2 \geq B(M)_{n,1}^2
\]
where $M$ is the adjacency or Laplacian matrix of the graph, depending on the choice of $H$, and
\[
\ket{\phi} = \frac{1}{\sqrt{p}} (\alpha\ket{0} + \beta e^{i\theta }B(M)_{n,1}\ket{1})
\]
for some phase factor $e^{i\theta}$. If $\ket{\phi'}$ is the state obtained from $\ket{\phi}$ by correcting for the phase $\theta$, then
\[
F(\ket{\psi},\ket{\phi'}) \geq \frac{2\sqrt{B_{n,1}(M)}}{1+ B_{n,1}(M)}.
\]
\end{lemma}
\begin{proof}
Since there is peak state transfer, the state after time $\tau$ is 
\[
\alpha\ket{0^n} + \beta e^{i\theta}B(M)_{n,1}\ket{n} + \beta \sum_{j=1}^{n-1} U(\tau)_{j,1}\ket{j},
\]
for some phase factor $e^{i\theta}$. Measuring the first $n-1$ qubits makes the state collapse to $\ket{0^{n-1}} \otimes \ket{\phi}$ with probability $p$. Correcting for the phase factor $e^{i\theta}$, we obtain the state $\ket{0^{n-1}} \otimes \ket{\phi'}$ where
\[
\ket{\phi'} = \frac{1}{\sqrt{p}} (\alpha\ket{0} + \beta B(M)_{n,1}\ket{1}).
\]
Writing $b \coloneqq B(M)_{n,1}$ for readability, we have that
\[
F(\ket{\psi},\ket{\phi'})^2 = \frac{(|\alpha|^2 + |\beta|^2b)^2}{|\alpha|^2 + |\beta|^2b^2}.
\]
Since $|\alpha|^2 = 1 - |\beta|^2 \in [0,1]$, we can lower bound $F(\ket{\psi},\ket{\phi'})^2$ by determining the minimum of
\[
f(x) \coloneqq \frac{(x + (1-x)b)^2}{x + (1-x)b^2}
\]
on the interval $[0,1]$. Taking the derivative with respect to $x$ yields a numerator that equals zero on this interval if and only if $x = b/(b+1)$; since $f$ is upper-bounded by $1$ and $f(0) = f(1) = 1$, the function takes on a minimum on (0,1) at this value of $x$, and so does $\sqrt{f}$, yielding
\[
F(\ket{\psi},\ket{\phi'}) \geq \sqrt{f(b/(b+1)}) = \frac{2\sqrt{b}}{1+b}. \qedhere 
\]
\end{proof}

\section{Spectral characterization}\label{sec:main}

Let $M$ be a symmetric matrix with  spectral decomposition $M = \sum_{r=0}^d \theta_r E_r$, indexed by some set $X$. The \textsl{eigenvalue support of $u \in X$} is 
\[
\ev_M(u) = \{\theta_r :  (E_r)_{u,u} \neq 0 \}  
\]
For $u,v \in X$, the  \textsl{positive  mutual eigenvalue support of $u$ and $v$}
\[
\ev^+_M(u,v) = \{\theta_r :  (E_r)_{v,u} > 0 \}  
\]
and the \textsl{negative mutual  eigenvalue support of $u$ and $v$} is 
\[
\ev^-_M(u,v) = \{\theta_r :  (E_r)_{v,u}  < 0 \}.
\]
The union of these two sets, which we denote by $\ev_M(u,v)$, is simply called the \textsl{mutual eigenvalue support of $u$ and $v$}; it consists of those eigenvalues for which the $(u,v)$-entry of the corresponding spectral idempotent is nonzero. We will omit the subscripts and write $\ev^{(\pm)}(u,v)$ when the context is clear. 

All of our results hold when we choose $M$ to be the adjacency or Laplacian matrix of a graph, but we will prove it in a more generalized setting, so that we can apply our results to certain weighted graphs. In this setting, it is enough to impose suitable algebraic conditions on the eigenvalues in the mutual eigenvalue support of $u$ and $v$.

\begin{definition}
Let $M$ be a real symmetric matrix and let $u,v$ index two of its rows and columns. We say that the triple $(M,u,v)$ is \emph{eligible} if every eigenvalue in $\ev_M(u,v)$ is an algebraic integer and $\ev_M(u,v)$ is closed under algebraic conjugation.
\end{definition}

Recall that a real number is an \emph{algebraic number} if it is a root of a  polynomial with integer coefficients, and an \emph{algebraic integer} if it is a root of a monic polynomial with integer coefficients. Two algebraic numbers are \emph{algebraic conjugates} if they are roots of the same irreducible polynomial over $\mathbb{Q}$. Thus, closure under algebraic conjugation means that whenever $\theta\in \ev_M(u,v)$, every algebraic conjugate of $\theta$ also belongs to $\ev_M(u,v)$.

We note that if the characteristic polynomial of $M$ has integer coefficients --- which holds, in particular, if the entries of $M$ are integers --- then it follows that $(M,u,v)$ is eligible for every pair of vertices $u,v$. The following appears under other guises in several sources, including \cite{CouGodBook}.

\begin{lemma}\label{lem:technical lemma}
    Let  $M$ be a symmetric $n\times n$ matrix and let $B=B(M)$ be the bounding matrix. Then 
    \[
    |U(\tau)_{v,u}|  \leq B(M)_{v,u},
    \]
    and equality holds (i.e.\ there is peak state transfer from $u$ to $v$ at time $\tau$) if and only if there exists a complex number $\gamma$ such that for $\theta_r\in \ev^{\pm}(u,v)$, we have 
\[
e^{it\theta_r} =  \pm \gamma. 
\]
\end{lemma}

\begin{proof}
    We will repeat the argument in \cref{sec:peak} and examine the equality cases. 
\begin{align*}
|U(t)_{v,u}| &= \left| \sum_{r=0}^d e^{it\theta_r} (E_r)_{v,u}\right|    \\
 &\leq   \sum_{r=0}^d \left| e^{it\theta_r} (E_r)_{v,u}\right| =
  B(M)_{v,u}.
\end{align*}
The inequality here comes from the triangle inequality and holds with equality if and only if the complex numbers 
\[e^{it\theta_0} (E_0)_{v,u}, \ldots, e^{it\theta_d} (E_d)_{v,u} \]
all have the same argument. Since the entries of $E_r$ are real, this holds if and only if there exists a complex number $\gamma$ such that 
\[
e^{it\theta_r} = \gamma, \quad \text{ for } \theta_r\in \ev^+(u,v)
\quad \text{ and } \quad 
e^{it\theta_r} = -\gamma, \quad \text{ for } \theta_r\in \ev^-(u,v)
\]
and the statement follows.
\end{proof}

\begin{corollary}
\label{cor:even_odd_multiples_of_pi}
Let $M$ be a symmetric $n\times n$ matrix and $u,v \in \{1,\ldots,n\}$ be distinct. Moreover, let $\theta_s$ be an eigenvalue in the positive mutual eigenvalue support of $u$ and $v$. There is peak state between $u$ and $v$ with respect to $M$ at time $\tau > 0$ if and only if the following conditions hold:
\begin{enumerate}[(i)]
    \item $\tau(\theta_s - \theta_r)$ is an even multiple of $\pi$ for all $\theta_r \in \ev^+(u,v)$; and
    \item $\tau(\theta_s - \theta_r)$ is an odd multiple of $\pi$ for all $\theta_r \in \ev^-(u,v)$.
\end{enumerate}
If these conditions hold, the peak state transfer occurs with phase factor $e^{i\tau\theta_s}$.
\end{corollary}
\begin{proof}
By Lemma \ref{lem:technical lemma}, there is peak state transfer between $u$ and $v$ at time $\tau$ if and only if 
\[
e^{i\tau\theta_r} = \pm e^{i\tau\theta_s} = \pm\gamma
\]
for all $\theta_r \in \ev^{\pm}(u,v)$, where $\gamma = e^{i\tau\theta_s}$. This happens if and only if $e^{i\tau(\theta_r - \theta_s)} = \pm 1$ whenever $\theta_r \in \ev^\pm(u,v)$ and the statement follows.
\end{proof}

A consequence of \ref{cor:even_odd_multiples_of_pi} is that for peak state transfer to occur, the following condition needs to hold: for any $\theta_x, \theta_y, \theta_w,\theta_z \in \ev(u,v)$ such that $\theta_w \neq \theta_z$, the number
\[
\frac{\theta_x - \theta_y}{\theta_w-\theta_z}
\]
is rational. This condition is called the \textsl{ratio condition} in the literature, see for example \cite{God2010}.  

In the following theorem, we give a characterization of peak state transfer with respect to some matrix $M$, provided that it satisfies some (weak) condition on its eigenvalues. It is implicitly assumed that the mutual eigenvalue support is non-empty so that there is no zero state transfer.
The following draws heavily on previous work and our only contribution here is to make the statement explicit and to lift to the more general setting where $M$ need not have integer entries. 

\begin{theorem}\label{thm:peakST-char-time}
Let $M$ be a symmetric $n \times n$ matrix and assume that $u$ and $v$ are distinct and such that the triple $(M,u,v)$ is eligible. There is peak state transfer between $u$ and $v$ if and only if the following conditions hold:
\begin{enumerate}[(i)]
    \item Every eigenvalue in $\ev(u,v)$ is a quadratic integer. In particular, there exists a squarefree integer $D \geq 1$, an integer $a$ and for each $\theta_r \in \ev(u,v)$ an integer $b_r$ having the same parity as $a$ such that
    \[
    \theta_r = \frac{1}{2}(a + b_r\sqrt{D})
    \]
    for all $\theta_r \in \ev(u,v)$.
    \item Let $s$ be such that $\theta_s \in \ev^+(u,v)$, and let $g = \gcd\{c_r : \theta_r \in \ev(u,v)\}$, where the $c_r$ are the integers given by
    \[
    c_r \coloneq \frac{\theta_s - \theta_r}{\sqrt{D}} = \frac{b_s - b_r}{2}.
    \]
    Then the following are true for all $\theta_r \in \ev(u,v)$:
    \begin{enumerate}[(a)]
        \item $c_r/g$ is even if and only if $\theta_r \in \ev^+(u,v)$, and
        \item $c_r/g$ is odd if and only if $\theta_r \in \ev^-(u,v)$.
    \end{enumerate}
\end{enumerate}
If these conditions hold, the peak state transfer occurs precisely at all times that are odd multiples $\tau$ of
\[
\tau_0 \coloneq \frac{\pi}{g\sqrt{D}}
\]
and the phase factor of peak state transfer is then $e^{i\tau \theta_s}$.
\end{theorem}

\begin{proof}
If conditions (i) and (ii) hold, then
\[
\tau_0(\theta_s - \theta_r) = \frac{c_r\pi}{g}
\]
is an even (odd) multiple of $\pi$ whenever $\theta_r$ is in the positive (negative) mutual eigenvalue support of $u$ and $v$, so there is peak state transfer at time $\tau_0$ by Corollary \ref{cor:even_odd_multiples_of_pi}. Now let $0 < \tau' \leq \tau_0$ denote the smallest time for which $\tau'(\theta_s - \theta_r)$ is an integer multiple of $\pi$. By construction, $\tau' = \tau_0$, which can be argued as follows. Let $m \geq 1$ be the largest integer such that $m\tau' \leq \tau_0$. Then $\tau_0 - m\tau' < \tau'$ and $(\tau_0 - m\tau')(\theta_s - \theta_r)$ is an integer multiple of $\pi$ for all $r$. By the minimality of $\tau'$ it must be the case that $\tau_0 = m\tau'$. Hence for all $r$, we have that $\tau'(\theta_s - \theta_r)/\pi = c_r/(mg)$ is an integer; since $g$ is the greatest common divisor of all the $c_r$, we conclude that $m = 1$. By a similarly flavoured argument, an arbitrary time $\tau$ is such that $\tau(\theta_s - \theta_r)$ is an integer multiple of $\pi$ for all $r$ if and only if $\tau$ is an integer multiple of $\tau_0$. In particular, the conditions of Corollary \ref{cor:even_odd_multiples_of_pi} are satisfied precisely when $\tau$ is an odd multiple of $\tau_0$.

We will now argue that condition (i) is necessary, so we assume that there is peak state transfer between $u$ and $v$ at some time $\tau$. If every eigenvalue in $\ev(u,v)$ is integer, the condition automatically holds with $a = 0$, $b_r = 2\theta_r$ and $D = 1$, so assume that $\theta_1$ is an eigenvalue in the mutual eigenvalue support that is not integer. Let $F$ denote the splitting field of the minimal polynomial $f$ of $\theta_1$ over $\rats$. Since $(M,u,v)$ is eligible, the set $S$ of roots of $f$ is fully contained in $\ev(u,v)$. Let $\theta_2 \neq \theta_1$ be another element of $S$. By the ratio condition, $(\theta_1 - \theta_2)/(\theta_r - \theta_q) \eqcolon \ell_{r,q}$ is rational for all pairs of distinct $\theta_r,\theta_q \in S$. Therefore, since
\[
(\theta_1 - \theta_2)^{|S|(|S|-1)} = \prod_{\substack{\theta_r,\theta_q \in S \\ r \neq q}} (\theta_r - \theta_q)\ell_{r,q}
\]
is invariant under Galois automorphisms of $F$ over $\rats$, it must be rational. In particular, it is an integer because $\theta_1$ and $\theta_2$ are algebraic integers. Let $m \geq 1$ be the smallest integer such that $(\theta_1 - \theta_2)^m \eqcolon q$ is integer, so that the minimal polynomial of $\theta_1 - \theta_2$ is $g \coloneq X^m - q \in \rats[X]$. Since $\theta_1 - \theta_2 \in F$, which is a normal extension of $\rats$, the splitting field of $g$ over $\rats$ is a subset of $F$, hence all roots of $g$. This implies that $m \leq 2$, as $g$ has at most two real roots. We conclude that $\theta_1 - \theta_2$ is the square root of some integer, i.e.\ it can be written in the form
\[
\theta_1 - \theta_2 = \alpha_{1,2} \sqrt{D_{1,2}}
\]
for a certain integer $\alpha_{1,2}$ and a squarefree integer $D_{1,2} \geq 1$. Let $\theta_r$ and $\theta_q$ be arbitrary distinct eigenvalues in $\ev(u,v)$ (not necessarily in $S$). By the ratio condition, there is a rational number $\ell_{r,q}$ such that $\theta_r - \theta_q = \ell_{r,q} (\theta_1 - \theta_2)$, hence
\[
(\theta_r - \theta_q)^2 = \ell_{r,q}^2 (\theta_1 - \theta_2)^2
\]
is an integer, which means that
\[
\theta_r - \theta_q = \alpha_{r,q} \sqrt{D_{r,q}}
\]
for a certain integer $\alpha_{r,q}$ and squarefree integer $D_{r,q}$. Then for any $\theta_r \neq \theta_q$ and $\theta_x \neq \theta_y$, we have that
\[
\alpha_{r,q}\alpha_{x,y}\sqrt{D_{r,q}D_{x,y}} = (\theta_r - \theta_q)(\theta_x - \theta_y) = \ell_{r,q}\ell_{x,y} (\theta_1 - \theta_2)^2
\]
is an algebraic integer in $\rats$, i.e.\ it is an integer. Because of the squarefreeness, this happens precisely when $D_{r,q} = D_{x,y}$. We conclude that there is a single squarefree integer $D$ such that $D_{r,q} = D$ for all pairs of distinct $r,q$. Moreover, we can write
\[
\sum_{\theta_r \in S} \theta_r = \sum_{r} \theta_1 - \alpha_{1,r}\sqrt{D} = |S|\theta_1 - \sqrt{D}\sum_{r} \alpha_{1,r},
\]
which is an integer, because the negative of left-hand side appears as the coefficient for $X^{|S| - 1}$ in $f$. Hence $\theta_1 \in \rats(\sqrt{D})$; note that since we assumed that $\theta_1$ is not integer, it must be that $D > 1$. More generally,
\[
\theta_r = \theta_1 - \alpha_{1,r}\sqrt{D} \in \rats(\sqrt{D})
\]
for any $\theta_r \in \ev(u,v)$. Due to a result of Dedekind for quadratic integers, we obtain that
\[
\theta_r = \frac{1}{2}(a_r + b_r\sqrt{D})
\]
for some pair of integers $a_r$ and $b_r$ that are either both even or both odd. Moreover, as $\theta_r - \theta_q$ is an integer multiple of $\sqrt{D}$, it must be that $a_r = a_q$ for all $r \neq q$, hence the condition (i) is satisfied for a unique $a$, with all the $b_r$ having the same parity as $a$.

Finally, let $\theta_s \in \ev^{+}(u,v)$; such a $\theta_s$ exists, since $\ev(u,v)$ is non-empty by assumption, and otherwise the $(u,v)$-entries of the idempotents $E_r$ cannot sum to zero.
Like in the statement of condition (ii), let $g$ be the greatest common divisor of the integers
\[
c_r = \frac{\theta_s - \theta_r}{\sqrt{D}} = \frac{b_s - b_r}{2}
\]
over all $\theta_r \in \ev(u,v)$. Since we assume that there is peak state transfer at time $\tau >0$, Corollary \ref{cor:even_odd_multiples_of_pi} tells us that for all $\theta_r \in \ev(u,v)$,
\[
\frac{\tau(\theta_r - \theta_s)}{\pi} = \frac{\tau\sqrt{D}}{\pi} \cdot c_r
\]
is an integer. Hence $\tau \sqrt{D}/\pi$ is a rational number, say, equal to $p/q$ for some integers $p$ and $q$ that are coprime. As $p c_r/q$ is integer, it must be that $q$ divides $c_r$ for all $\theta_r \in \ev(u,v)$, hence $g = mq$ for some integer $m$. Now the parity of
\[
\frac{\tau(\theta_s - \theta_r)}{\pi} = pm\cdot \frac{c_r}{g}
\]
depends on whether $\theta_r$ is in the positive or negative mutual eigenvalue support of $u$ and $v$. As $\ev^-(u,v)$ is nonempty, $pm$ is necessarily odd, hence the conditions (i) and (ii) in Corollary \ref{cor:even_odd_multiples_of_pi} correspond precisely to the conditions in (a) and (b) in the statement of this theorem.
\end{proof}

If condition (i) of Theorem \ref{thm:peakST-char-time} holds and if some $\theta_r \in \ev_M(u,v)$ is an integer, then $D = 1$ and all elements of $\ev_M(u,v)$ are integer. In particular, the Laplacian matrix of any (weighted) graph always has $0$ as an eigenvalue. If the graph is connected, this eigenvalue has multiplicity $1$ and the corresponding spectral idempotent is a scalar multiple of the all-ones matrix. Hence $0$ is always in the positive mutual eigenvalue support of any pair of vertices and we can take $\theta_s = 0$ in the theorem above to obtain the following corollary:
\begin{corollary}
Let $L$ be the Laplacian matrix of a connected graph $G$ with algebraic integer weights. There is peak state transfer between a pair of distinct vertices $u$ and $v$ of $G$ with respect to $L$ if and only if the following conditions hold:
\begin{enumerate}[(i)]
\item Every $\theta_r \in \ev_L(u,v)$ is an integer.
\item Let $g = \gcd\left(\ev_L(u,v)\right)$. The following are true for all $\theta_r \in \ev_L(u,v)$:
\begin{enumerate}[(a)]
\item $\theta_r/g$ is even if and only if $\theta_r \in \ev_L^+(u,v)$;
\item $\theta_r/g$ is odd if and only if $\theta_r \in \ev_L^-(u,v)$.
\end{enumerate}
\end{enumerate}
If these conditions hold, the peak state transfer occurs precisely at all times that are odd multiples $\tau$ of $\tau_0 = \pi/g$ and the phase factor of the peak state transfer is $1$.
\end{corollary}

\begin{remark}
If there is peak state transfer between $u$ and $v$ at time $\tau$ with respect to $M$, we know from Lemma \ref{lem:technical lemma} that there is a $\gamma$ such that $e^{it\theta_r} = \pm \gamma$ for every $\theta_r$ in the mutual eigenvalue support. This means that $e^{2i\tau\theta_r} = \gamma^2$ for every such $\theta_r$. Hence
\[
|U(2\tau)_{v,u}| = \left| \sum_{r=0}^d e^{2i\tau\theta_r} (E_{r})_{v,u}\right| = \gamma^2 \left|\sum_{r=0}^d (E_r)_{v,u}\right| = 0,
\]
so that the transfer probability from $u$ to $v$ at time $2\tau$ is 0.

In this section, we have required $u,v$ to be distinct vertices, because this is the interesting case. We could have also allowed $u=v$ (or peak state transfer to itself), but it can be easily seen that this is equivalent to \textsl{periodicity} (which occurs when a diagonal entry of $|U(t)|$ equals $1$) since the diagonal entries of the bounding matrix are always $1$.

If there is peak state transfer from $u$ to $v$ at time $\tau$, then the walk is not necessarily periodic at $u$ at time $2\tau$, as it would be in the case of perfect state transfer. After all, $\ev(u,u)$ may contain eigenvalues that are not in $\ev(u,v)$, possibly invalidating the conditions in Theorem \ref{thm:peakST-char-time}.
\end{remark}

\section{Examples of peak state transfer}\label{sec:exs}

In this section, we give many examples of peak state transfer. First we give sporadic examples in \cref{sec:sporadicexamples}. In \cref{sec:fam-going-to-1}, we give an infinite family of graphs with peak state transfer where the transfer amount goes to $1$ as the number of vertices increases. This shows that we can have peak state transfer coming arbitrarily close to $1$, but between vertices at distance $2$. On the other hand, in \cref{sec:fam-high-peak-wtd-path}, we give an infinite family of weighted paths where the peak state transfer goes between vertices arbitrarily far away.

\subsection{Small sporadic examples}\label{sec:sporadicexamples}

In this section, we have three illustrative examples of peak state transfer in small graphs. Graphs with integer eigenvalues (with respect to the adjacency matrix) have good potential to have peak state transfer. The Petersen graph is such an example. We also give an example of peak state transfer between vertices at distance $4$ (with respect to the adjacency matrix) and an example of peak state transfer with respect to the Laplacian matrix.

\subsubsection*{The Petersen graph}

The Petersen graph \cite{HolShe1993} is a well-known graph in graph theory, which has many symmetries and  forms a  counterexample to many would-be conjectures. Although it does not admit perfect state transfer, it does admit peak state transfer (with respect to the adjacency matrix) at time $\pi$ between every pair of adjacent vertices, where the transfer probability is $64/225 \approx 0.284444$. See Figure \ref{fig:pete-peak}. The lower bound on the fidelity, following Lemma \ref{lem:fidelity} is $(60/23)\sqrt{2/15} \approx 0.952561$. 

\begin{figure}[htbp]
  \centering
  \begin{minipage}{0.4\textwidth}
    \centering
\begin{tikzpicture}[scale=0.4, every node/.style={scale=1.0}]
\tikzstyle{graph node}=[circle, fill=black, draw=black, inner sep=1.5pt]
 
\useasboundingbox (-5.05,-4.4) rectangle (5.1,5.25);
 
\begin{scope}[rotate=90]
 \foreach \x/\y in {0/1,72/2,144/3,216/4,288/5}{
 \node[graph node] (\y) at (canvas polar cs: radius=2.5cm,angle=\x){};
 }
 \foreach \x/\y in {0/6,72/7,144/8,216/9,288/10}{
 \node[graph node] (\y) at (canvas polar cs: radius=5cm,angle=\x){};
 }
\end{scope}
 
\foreach \x/\y in {1/6,2/7,3/8,4/9,5/10}{
 \draw (\x) -- (\y);
}
 
\foreach \x/\y in {1/3,2/4,3/5,4/1,5/2}{
 \draw (\x) -- (\y);
}
 
\foreach \x/\y in {6/7,7/8,8/9,9/10,10/6}{
 \draw (\x) -- (\y);
}
\node[left] at (1) {$u$};
\node[left] at (6.north) {$v$};
\end{tikzpicture}
    \par\vspace{1ex}
    \textbf{(a)} the Petersen graph
  \end{minipage}
  \quad
  \begin{minipage}{0.55\textwidth}
    \centering
    \includegraphics[width=\textwidth]{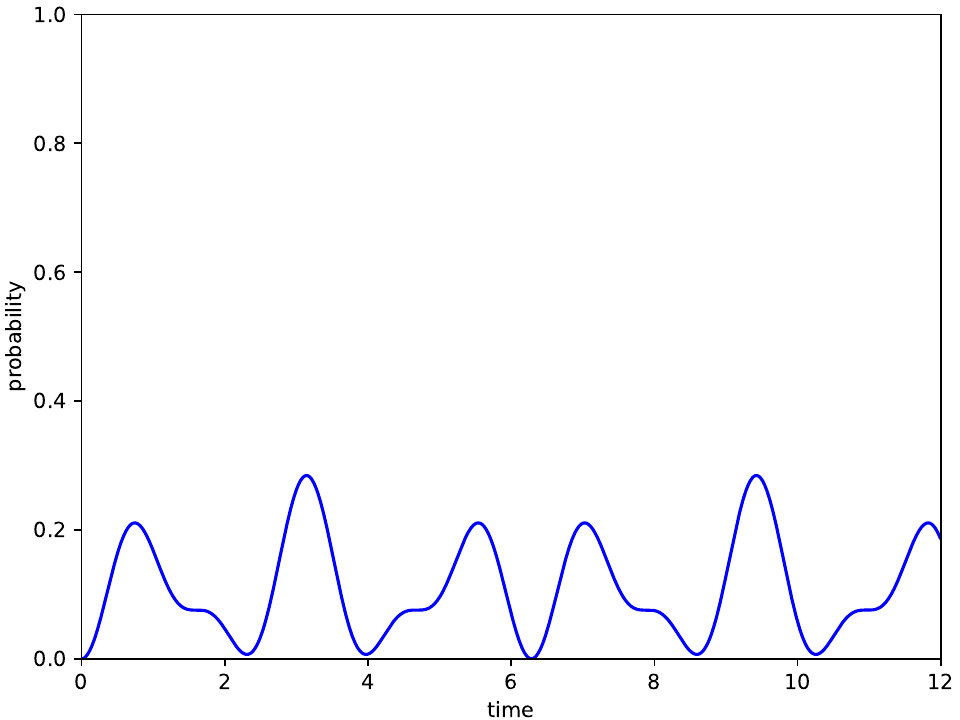}
    \par\vspace{1ex}
    \textbf{(b)} $(U(t)\circ \comp{U(t)})_{u,v}$
  \end{minipage}
  
  \caption{The Petersen graph and the probability of measuring at $v$ when starting at vertex $u$, for any pair of adjacent vertices $u$ and $v$.}
  \label{fig:pete-peak}
\end{figure}

\subsubsection*{Peak state transfer in a small planar graph}

The 11-vertex planar graph $G_{11}$ depicted in Figure \ref{fig:small_planar_peakST2} admits peak state transfer with respect to the adjacency matrix between vertex $u$ and vertex $v$ at time $\pi$ over distance $4$; the corresponding transfer probability at that time is $16/25 = 0.64$. 
The lower bound on the fidelity given by Lemma \ref{lem:fidelity} is $(20/9)\sqrt{1/5} \approx 0.993808$.

\begin{figure}[H]
  \centering
  \begin{minipage}{0.4\textwidth}
    \centering
\begin{tikzpicture}[scale=0.5, every node/.style={scale=1.0}]
\tikzstyle{graph node}=[circle, fill=black, draw=black, inner sep=1.5pt]

\node[graph node] (0) at (-2,2.5) {};
\node[graph node] (1) at (-2,-2.5) {};
\node[graph node] (2) at (-1,0) {};
\node[graph node] (3) at (1,0) {};
\node[graph node] (4) at (3,0) {};
\node[graph node] (5) at (5,2.5) {};
\node[graph node] (6) at (5,1) {};
\node[graph node] (7) at (5,-1) {};
\node[graph node] (8) at (5,-2.5) {};
\node[graph node] (9) at (6.5,0) {};
\node[graph node] (10) at (8,0) {};

\node[left] at (2.west) {$u$};

\foreach \x/\y in {0/1,0/2,0/3,1/2,1/3,2/3,0/5,1/8,3/4,4/6,4/7,5/6,6/7,7/8,6/9,7/9,5/10,6/10,7/10,8/10,9/10}{
    \draw (\x) -- (\y);
}
\node[left] at (9.west) {$v$};

\end{tikzpicture}
    \par\vspace{1ex}
    \textbf{(a)} $G_{11}$
  \end{minipage}
  \quad
  \begin{minipage}{0.55\textwidth}
    \centering
    \includegraphics[width=\textwidth]{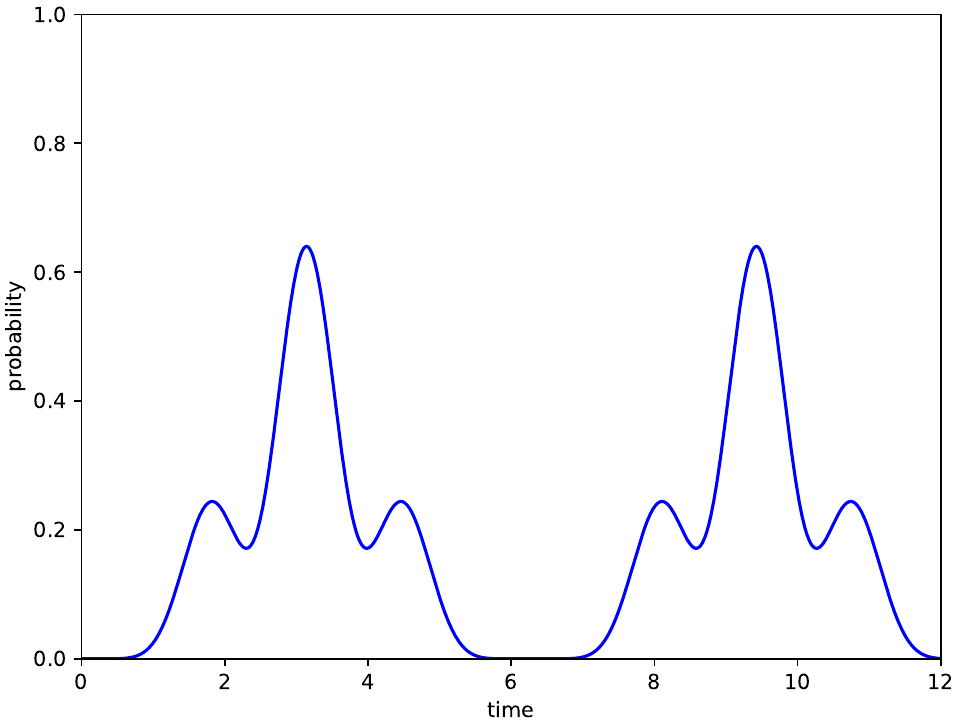}
    \par\vspace{1ex}
    \textbf{(b)} $(U(t)\circ \comp{U(t)})_{u,v}$
  \end{minipage}
  
  \caption{$G_{11}$ and the probability of measuring at $v$ when starting at vertex $u$. 
  \label{fig:small_planar_peakST2}}
\end{figure}

\subsubsection*{Laplacian peak state transfer}

The 12-vertex graph $G_{12}$ depicted in Figure \ref{fig:small_laplacian_peakST} admits peak state transfer from vertex $u$ to vertex $v$ with respect to its Laplacian matrix. The time of peak state transfer is $\pi/2$ and the transfer probability is $4/9 \approx 0.444444$. This graph is planar and is isomorphic to its planar dual. The lower bound on the fidelity given by Lemma \ref{lem:fidelity} is $(6/5)\sqrt{2/3} \approx 0.979796$. 

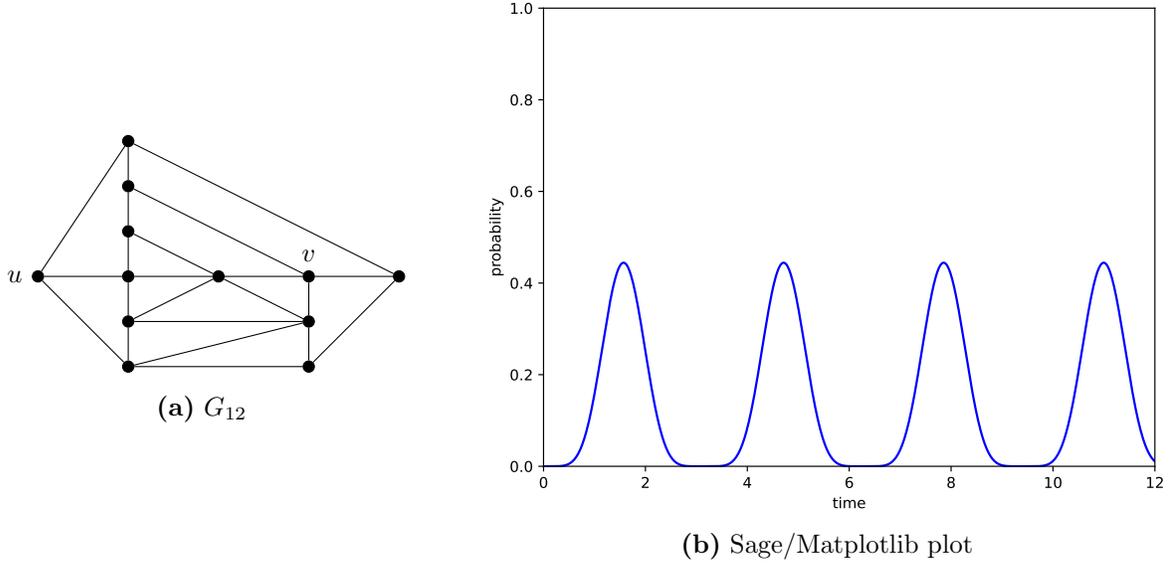
\begin{figure}[H]
  \centering
  \begin{minipage}{0.4\textwidth}
    \centering
\begin{tikzpicture}[scale=0.6, every node/.style={scale=1.0}]
\tikzstyle{graph node}=[circle, fill=black, draw=black, inner sep=1.5pt]
\node[graph node] (0) at (-2,0) {};
\node[graph node] (1) at (0,3) {};
\node[graph node] (2) at (0,2) {};
\node[graph node] (3) at (0,1) {};
\node[graph node] (4) at (0,0) {};
\node[graph node] (5) at (0,-1) {};
\node[graph node] (6) at (0,-2) {};
\node[graph node] (7) at (2,0) {};
\node[graph node] (8) at (4,0) {};
\node[graph node] (9) at (4,-1) {};
\node[graph node] (10) at (4,-2) {};
\node[graph node] (11) at (6,0) {};

\node[left] at (0.west) {$u$};
\node[above] at (8.north) {$v$};

\foreach \x/\y in {0/1,0/4,0/6,1/2,2/3,3/4,4/5,5/6,1/11,2/8,3/7,4/7,5/7,5/9,6/9,6/10,7/8,7/9,8/9,9/10,8/11,10/11}{
    \draw (\x) -- (\y);
}

\end{tikzpicture}
    \par\vspace{1ex}
    \textbf{(a)} $G_{12}$
  \end{minipage}
  \quad
  \begin{minipage}{0.55\textwidth}
    \centering
    \includegraphics[width=\textwidth]{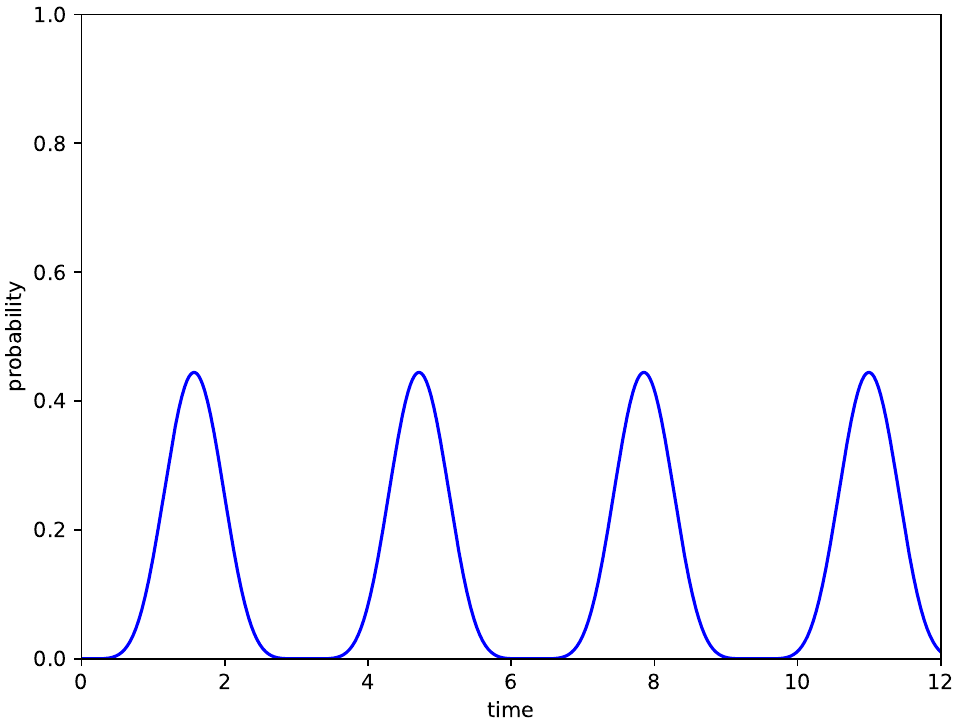}
    \par\vspace{1ex}
    \textbf{(b)} $(U(t)\circ \comp{U(t)})_{u,v}$
  \end{minipage}
  
  \caption{$G_{12}$ and the probability of measuring at $v$ when starting at vertex $u$. 
  \label{fig:small_laplacian_peakST}}
\end{figure}

\subsection{Family of graphs with peak state transfer going to \texorpdfstring{$1$}{1}}\label{sec:fam-going-to-1}

In this section, we given a family of graphs with a pair of vertices having peak state transfer with probability tending to $1$ as the number of vertices increases. We note that the corresponding fidelity as given by Lemma 4.1 also tends to $1$.

Consider the following family of graphs. For some positive integer $n$, we take $n$ copies of $K_2$ and then we add a pair of vertices $u$ and $v$ and add for each copy of $K_2$ a pair of edges going from one of its endpoints to $u$ and $v$. We denote the resulting graph by $G_n$. For $n=1,\ldots, 4$ the corresponding graphs are depicted in Figure \ref{fig:peakST_infinite_family}.

\begin{figure}[htbp]
\centering
\begin{tikzpicture}[scale=1, every node/.style={scale=1.0}]
\tikzstyle{graph node}=[circle, fill=black, draw=black, inner sep=1.5pt]
    \node (gn) at (0,2) {$G_1$};
    
    \node[graph node] (0) at (0.5,-1) {};
    \node[graph node] (1) at (-0.5,-1) {};
    \node[graph node] (2) at (0,0) {};
    \node[graph node] (3) at (0,1) {};

    \node[below] at (0.south) {$v$};
    \node[below] at (1.south) {$u$};
    
    \foreach \x/\y in {1/2,0/2,2/3}{
        \draw (\x) -- (\y);
    }
\end{tikzpicture}\qquad\quad
\begin{tikzpicture}[scale=1, every node/.style={scale=1.0}]
\tikzstyle{graph node}=[circle, fill=black, draw=black, inner sep=1.5pt]
    \node (gn) at (0,2) {$G_2$};
    
    \node[graph node] (0) at (0.5,-1) {};
    \node[graph node] (1) at (-0.5,-1) {};
	\node[graph node] (2) at (-.5, 0) {};
	\node[graph node] (3) at (-.5, 1) {};
	\node[graph node] (4) at (.5, 0) {};
	\node[graph node] (5) at (.5, 1) {};
    
    \node[below] at (0.south) {$v$};
    \node[below] at (1.south) {$u$};
    
    \foreach \x/\y in {1/2,0/2,2/3,0/4,1/4,4/5}{
        \draw (\x) -- (\y);
    }
\end{tikzpicture}\qquad\quad
\begin{tikzpicture}[scale=1, every node/.style={scale=1.0}]
\tikzstyle{graph node}=[circle, fill=black, draw=black, inner sep=1.5pt]
    \node (gn) at (0,2) {$G_3$};
    
    \node[graph node] (0) at (0.5,-1) {};
    \node[graph node] (1) at (-0.5,-1) {};
	\node[graph node] (2) at (-.75, 0) {};
	\node[graph node] (3) at (-.75, 1) {};
	\node[graph node] (4) at (0, 0) {};
	\node[graph node] (5) at (0, 1) {};
    \node[graph node] (6) at (.75, 0) {};
	\node[graph node] (7) at (.75, 1) {};
    
    \node[below] at (0.south) {$v$};
    \node[below] at (1.south) {$u$};
    
    \foreach \x/\y in {1/2,0/2,2/3,0/4,1/4,4/5,0/6,1/6,6/7}{
        \draw (\x) -- (\y);
    }
\end{tikzpicture}\qquad\quad
\begin{tikzpicture}[scale=1, every node/.style={scale=1.0}]
\tikzstyle{graph node}=[circle, fill=black, draw=black, inner sep=1.5pt]
    \node (gn) at (0,2) {$G_4$};
    
    \node[graph node] (0) at (0.5,-1) {};
    \node[graph node] (1) at (-0.5,-1) {};
	\node[graph node] (2) at (-1, 0) {};
	\node[graph node] (3) at (-1, 1) {};
	\node[graph node] (4) at (-.3333, 0) {};
	\node[graph node] (5) at (-.3333, 1) {};
    \node[graph node] (6) at (.3333, 0) {};
	\node[graph node] (7) at (.3333, 1) {};
    \node[graph node] (8) at (1, 0) {};
	\node[graph node] (9) at (1, 1) {};
    
    \node[below] at (0.south) {$v$};
    \node[below] at (1.south) {$u$};
    
    \foreach \x/\y in {1/2,0/2,2/3,0/4,1/4,4/5,0/6,1/6,6/7,0/8,1/8,8/9}{
        \draw (\x) -- (\y);
    }
\end{tikzpicture}
\caption{The first four graphs in the infinite family of graphs for which the transfer probability of peak state transfer between $u$ and $v$ converges to $1$.\label{fig:peakST_infinite_family}}
\end{figure}
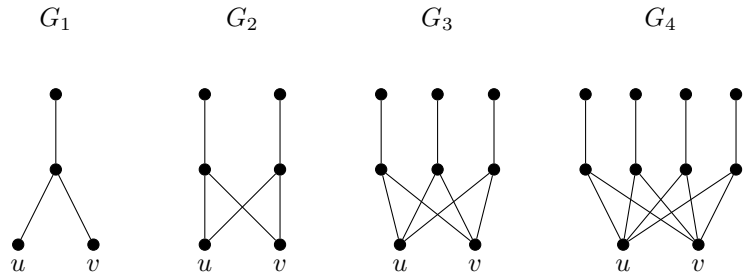

The adjacency matrix $A_n$ of $G_n$ can be written down as the block matrix
\[
\left[
\begin{array}{ccc}
0 & J_{2 \times n} & 0 \\
 J_{n \times 2} & 0 & I_n \\
 0 & I_n & 0
\end{array}
\right],
\]
the blocks of entries corresponding to $\{u,v\}$, the neighbours of $u$ and $v$ and the non-neighbours of $u$ and $v$, respectively (with respective block sizes $2$, $n$ and $n$). The eigenvalues of $A$ are $\pm \sqrt{2n+1}$, $\pm 1$ and $0$. The corresponding spectral idempotents are
\[
E_{\pm \sqrt{2n+1}} = \frac{1}{4n^2+2n}
\left[
\begin{array}{ccc}
n^2 J_{2\times 2} & \pm n\sqrt{2n+1} J_{2 \times n} & n J_{2\times n} \\
\pm n\sqrt{2n+1} J_{n \times 2} & (2n+1)J_{n\times n} & \pm \sqrt{2n+1} J_{n \times n} \\
n J_{n\times 2} & \pm \sqrt{2n+1} J_{n \times n} & J_{n\times n}
\end{array}
\right],
\]
\[
E_{\pm 1} = \frac{1}{2n}
\left[
\begin{array}{ccc}
0 & 0 & 0\\
0 & nI_n - J_{n \times n} & \pm (nI_n - J_{n \times n})\\
0 & \pm (nI_n - J_{n \times n}) & nI_n - J_{n \times n}
\end{array}
\right]
\]
and
\[
E_0 = \frac{1}{2n+1}
\left[
\begin{array}{ccc}
(2n+1)I_{2} - nJ_{2 \times 2} & 0 & -nJ_{2 \times n}\\
0 & 0 & 0 \\
 -J_{n \times 2} & 0 & 2 J_{n \times n}
\end{array}
\right]
\]
by a straightforward calculation.
We find that the mutual eigenvalue support of $u$ and $v$ consists of the eigenvalues $\pm\sqrt{2n+1}$ and $0$, with the positive and negative parts being
\[
\ev^+(u,v) = \{\pm \sqrt{2n+1} \} \quad \text{and} \quad \ev^-(u,v) = \{0\}.
\]
To verify the conditions of Theorem \ref{thm:peakST-char-time}, we can take $D$ to be the square-free part of $2n+1$. Since
\[
g = \gcd\left\{ 0, \pm \frac{\sqrt{2n+1}}{\sqrt{D}}  \right\} = \frac{\sqrt{2n+1}}{\sqrt{D}},
\]
it follows that condition (ii) of Theorem \ref{thm:peakST-char-time} is also verified, so that there is peak state transfer at all times that are odd multiples of
\[
\tau_0 = \frac{\pi}{g\sqrt{D}} = \frac{\pi}{\sqrt{2n+1}}.
\]
If we consider the $(u,v)$-entry of the bounding matrix $B(M)$, we see that
\[
B(M)_{u,v} = 2\cdot \frac{n^2}{4n^2 + 2n} + \frac{n}{2n+1} = \frac{2n}{2n+1},
\]
which converges to $1$ as $n$ increases.

\subsection{High peak state transfer in  weighted path graphs}\label{sec:fam-high-peak-wtd-path}

We define an infinite family of weighted paths and show that they admit peak state transfer with a high amount of transfer, nearly from end to end; the peak state transfer will go from the second vertex to the last vertex.

The weighted graph $X_n$ is defined as follows: the vertex set of $X_n$ is $v_1,\ldots, v_{n}$ and $w_1,\ldots, w_n$, and the edges $v_jw_{n+1-j}$ with weight $\sqrt{j(n+1-j)}$ for $j=1,\ldots, n$ and $w_j v_{n-j}$ with weight $\sqrt{j(n-j)}$ for $j=1,\ldots, n-1$. In other words, $X_n$ is a path on $2n$ vertices which visits the vertices in the following order:
$v_n,w_1, v_{n-1}, w_2, \ldots, v_2, w_{n-1}, v_1, w_n$
and the weight of an edge $v_jw_k$ is $\sqrt{jk}$. See Figure \ref{fig:weirdpaththings} for $X_2,X_3$ and $X_4$. 

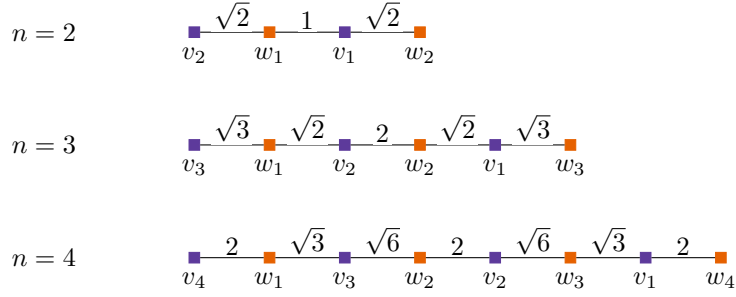
\begin{figure}[htbp]
    \centering
\begin{tikzpicture}[scale=1.0, every node/.style={scale=1.0}]
% Define styles
\tikzstyle{graph node}=[circle, fill=black, draw=black, inner sep=2pt]
\tikzstyle{path node v}=[rectangle, fill=kcolour, draw=kcolour, inner sep=2pt]
\tikzstyle{path node w}=[rectangle, fill=vcolour, draw=vcolour, inner sep=2pt]
\tikzstyle{edge label}=[midway, fill=white, inner sep=1pt]

\node at (-2,0) {$n=2$}; 
% Draw graph vertices
\node[path node v] (v1) at (0,0) {};
\node[path node w] (v2) at (1,0) {};
\node[path node v] (v3) at (2,0) {};
\node[path node w] (v4) at (3,0) {};
\draw (v1) -- node[edge label, above] {$\sqrt{2}$} (v2);
\draw (v2) -- node[edge label, above] {$1$} (v3);
\draw (v3) -- node[edge label, above] {$\sqrt{2}$} (v4);
\node[below] at (v1.south) {$v_2$};
\node[below] at (v2.south) {$w_1$};
\node[below] at (v3.south) {$v_1$};
\node[below] at (v4.south) {$w_2$};

\node at (-2,-1.5) {$n=3$}; 
% Draw graph vertices
\node[path node v] (x1) at (0,-1.5) {};
\node[path node w] (x2) at (1,-1.5) {};
\node[path node v] (x3) at (2,-1.5) {};
\node[path node w] (x4) at (3,-1.5) {};
\node[path node v] (x5) at (4,-1.5) {};
\node[path node w] (x6) at (5,-1.5) {};
\draw (x1) -- node[edge label, above] {$\sqrt{3}$} (x2);
\draw (x2) -- node[edge label, above] {$\sqrt{2}$} (x3);
\draw (x3) -- node[edge label, above] {$2$} (x4);
\draw (x4) -- node[edge label, above] {$\sqrt{2}$} (x5);
\draw (x5) -- node[edge label, above] {$\sqrt{3}$} (x6);
\node[below] at (x1.south) {$v_3$};
\node[below] at (x2.south) {$w_1$};
\node[below] at (x3.south) {$v_2$};
\node[below] at (x4.south) {$w_2$};
\node[below] at (x5.south) {$v_1$};
\node[below] at (x6.south) {$w_3$};

\node at (-2,-3) {$n=4$}; 
% Draw graph vertices
\node[path node v] (y1) at (0,-3) {};
\node[path node w] (y2) at (1,-3) {};
\node[path node v] (y3) at (2,-3) {};
\node[path node w] (y4) at (3,-3) {};
\node[path node v] (y5) at (4,-3) {};
\node[path node w] (y6) at (5,-3) {};
\node[path node v] (y7) at (6,-3) {};
\node[path node w] (y8) at (7,-3) {};
\draw (y1) -- node[edge label, above] {$2$} (y2);
\draw (y2) -- node[edge label, above] {$\sqrt{3}$} (y3);
\draw (y3) -- node[edge label, above] {$\sqrt{6}$} (y4);
\draw (y4) -- node[edge label, above] {$2$} (y5);
\draw (y5) -- node[edge label, above] {$\sqrt{6}$} (y6);
\draw (y6) -- node[edge label, above] {$\sqrt{3}$} (y7);
\draw (y7) -- node[edge label, above] {$2$} (y8);
\node[below] at (y1.south) {$v_4$};
\node[below] at (y2.south) {$w_1$};
\node[below] at (y3.south) {$v_3$};
\node[below] at (y4.south) {$w_2$};
\node[below] at (y5.south) {$v_2$};
\node[below] at (y6.south) {$w_3$};
\node[below] at (y7.south) {$v_1$};
\node[below] at (y8.south) {$w_4$};
\end{tikzpicture}
    \caption{The weighted paths $X_n$ for $n=2,3,4$.    \label{fig:weirdpaththings}}
\end{figure}

The adjacency matrices of these weighted paths are the normalized quotient matrices of the natural equitable partition of the graph $Y_n$, obtained from $X_n$ by replacing each $v_j,w_j$ with a set of vertices $V_j, W_j$ (resp.) of size $j$ and adding all edges between $V_j$ and $W_k$ if $v_j$ is adjacent to $w_k$. 

\begin{figure}[htbp]
    \centering
\begin{tikzpicture}[scale=1.0, every node/.style={scale=1.0}]
% Define styles
\tikzstyle{graph node}=[circle, fill=black, draw=black, inner sep=1.5pt]
\tikzstyle{path node v}=[rectangle, fill=kcolour, draw=kcolour, inner sep=2pt]
\tikzstyle{path node w}=[rectangle, fill=vcolour, draw=vcolour, inner sep=2pt]
\tikzstyle{edge label}=[midway, fill=white, inner sep=1pt]

\node[graph node] (x11) at (0,0) {};
\node[graph node] (x12) at (0,-0.7) {};
\node[graph node] (x13) at (0,-1.4) {};
\node[graph node] (x14) at (0,-2.1) {};

\node[graph node] (x21) at (1,-1.05) {};

\node[graph node] (x31) at (2,-0.35) {};
\node[graph node] (x32) at (2,-1.05) {};
\node[graph node] (x33) at (2,-1.75) {};

\node[graph node] (x41) at (3,-0.7) {};
\node[graph node] (x42) at (3,-1.4) {};

\node[graph node] (x51) at (4,-0.7) {};
\node[graph node] (x52) at (4,-1.4) {};

\node[graph node] (x61) at (5,-0.35) {};
\node[graph node] (x62) at (5,-1.05) {};
\node[graph node] (x63) at (5,-1.75) {};

\node[graph node] (x71) at (6,-1.05) {};

\node[graph node] (x81) at (7,0) {};
\node[graph node] (x82) at (7,-0.7) {};
\node[graph node] (x83) at (7,-1.4) {};
\node[graph node] (x84) at (7,-2.1) {};

\draw (x11) -- (x21);
\draw (x12) -- (x21);
\draw (x13) -- (x21);
\draw (x14) -- (x21);
\draw (x21) -- (x31);
\draw (x21) -- (x32);
\draw (x21) -- (x33);
\draw (x41) -- (x31);
\draw (x41) -- (x32);
\draw (x41) -- (x33);
\draw (x42) -- (x31);
\draw (x42) -- (x32);
\draw (x42) -- (x33);
\draw (x41) -- (x51);
\draw (x42) -- (x51);
\draw (x41) -- (x52);
\draw (x42) -- (x52);
\draw (x51) -- (x61);
\draw (x51) -- (x62);
\draw (x51) -- (x63);
\draw (x52) -- (x61);
\draw (x52) -- (x62);
\draw (x52) -- (x63);
\draw (x71) -- (x61);
\draw (x71) -- (x62);
\draw (x71) -- (x63);
\draw (x71) -- (x81);
\draw (x71) -- (x82);
\draw (x71) -- (x83);
\draw (x71) -- (x84);
% Draw graph vertices
\node[path node v] (y1) at (0,-3) {};
\node[path node w] (y2) at (1,-3) {};
\node[path node v] (y3) at (2,-3) {};
\node[path node w] (y4) at (3,-3) {};
\node[path node v] (y5) at (4,-3) {};
\node[path node w] (y6) at (5,-3) {};
\node[path node v] (y7) at (6,-3) {};
\node[path node w] (y8) at (7,-3) {};
\draw (y1) -- node[edge label, above] {$2$} (y2);
\draw (y2) -- node[edge label, above] {$\sqrt{3}$} (y3);
\draw (y3) -- node[edge label, above] {$\sqrt{6}$} (y4);
\draw (y4) -- node[edge label, above] {$2$} (y5);
\draw (y5) -- node[edge label, above] {$\sqrt{6}$} (y6);
\draw (y6) -- node[edge label, above] {$\sqrt{3}$} (y7);
\draw (y7) -- node[edge label, above] {$2$} (y8);
\node[below] at (y1.south) {$v_4$};
\node[below] at (y2.south) {$w_1$};
\node[below] at (y3.south) {$v_3$};
\node[below] at (y4.south) {$w_2$};
\node[below] at (y5.south) {$v_2$};
\node[below] at (y6.south) {$w_3$};
\node[below] at (y7.south) {$v_1$};
\node[below] at (y8.south) {$w_4$};
\end{tikzpicture}
    \caption{The graph $Y_4$ with its quotient $X_4$.  \label{fig:weirdpaththings2}}
\end{figure}
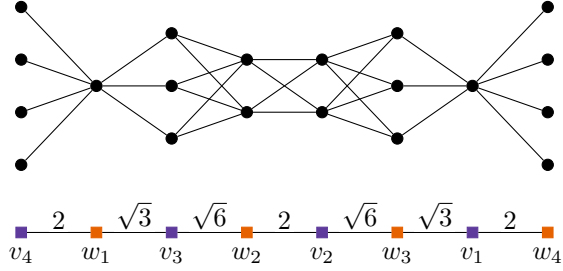

It turns out that $X_n$ is an integral weighted graph (i.e. the eigenvalues of $A_n$ are integers). Despite the simple structure of the spectrum of $A_n$, its eigenvectors do not seem to have simple closed expressions. Our proof of the following lemma works with the characteristic polynomial of $A_n$, avoiding the eigenvectors completely. To better preserve the flow of the paper, we defer the proof to the appendix. 

\begin{lemma}
    \label{lem:weighted_adj_eigenvalues}
    The eigenvalues of the weighted adjacency matrix $A_n$ of $X_n$ are $\pm 1, \pm 2, \ldots, \pm n$, each with multiplicity $1$.
\end{lemma}

\begin{proof}
    See Appendix \ref{app:proofs}.
\end{proof}

In the following Theorem, we show that peak state transfer occurs in $X_n$, almost from end to end, and we give a tight, constant lower bound on the transfer probability. In order to apply Theorem \ref{thm:peakST-char-time}, we need access to the $(w_1,w_n)$-entries of the spectral idempotents of $A_n$. For this, we can make use of standard tool that is used in for example in \cite{GodGuoKem2020}, which can be found in \cite{GodsilAlgebraicCombinatorics}.

Given the spectral decomposition
\[
M = \sum_{r=0}^d \theta_r E_r,
\]
of a symmetric matrix $M$, we can write
\[
(tI - M)^{-1} = \sum_{r=0}^d (t - \theta_r)^{-1}E_r
\]
so that for any $j \in \{1,\ldots,r\}$:
\[
(t-\theta_j)(tI - M)^{-1} = E_j + \sum_{r \neq j} \frac{t - \theta_r}{t - \theta_r} E_r.
\]
This means that the $(u,v)$-entry of $E_j$ can be expressed as follows:
\[
(E_j)_{uv} = \lim_{t \to \theta_j} (t - \theta_r)(tI - M)^{-1}_{uv}.
\]
The $(u,v)$-entry of the inverse of $(tI - M)$ is given by $\phi_{uv}(t)/\phi_M(t)$, where $\phi_{uv}(t)$ denotes the cofactor of $tI - M$ corresponding to the $uv$-entry, and where $\phi_{M}(t)$ denotes the characteristic polynomial of $M$. If $M = A$ is the adjacency matrix of a weighted graph $G$, the polynomial $\phi_{uv}$ can be written as
\[
\sum_{P:u \to v} \omega(P)\phi(G \setminus P,t),
\]
where the sum is over all paths $P$ in $G$ from $u$ to $v$, where $\omega(P)$ is the product of the edge weights in the path $P$ and where $\phi(G \setminus P,t)$ is the characteristic polynomial of the graph obtained by deleting the vertices in $P$ from $G$. (See Corollary 2.2 of Chapter 4 in \cite{GodsilAlgebraicCombinatorics}.) Putting this all together, we conclude that
\begin{equation}
\label{eq:idempotent_entry}
(E_j)_{uv} = \lim_{t \to \theta_j} \frac{(t-\theta_r)\sum_{P:u \to v}\omega(P)\phi(G\setminus P,t)}{\phi(G)}.
\end{equation}

\begin{theorem}
\label{thm:weighted_path_peakST}
The weighted graph $X_n$ admits peak state transfer, with respect to its adjacency matrix $A_n$, from $w_1$ to $w_n$, at time $\pi$. 
The sequence $(B(A_n)_{w_1,w_{n}})_{n=1}^{\infty}$ is decreasing and converges to $\sqrt{\pi}/2$. Consequently, if the walk starts at $w_1$, the probability to measure at $w_n$ at time $\pi$ converges to (and is at least) $\pi/4 \approx 0.785398$.
\end{theorem}

\begin{proof}
In this proof, we will index the spectral idempotents of $A_n$ by its eigenvalues, i.e., we will write 
\[
A_n = \sum_{k \in \{\pm 1,\ldots,\pm n\}} kE_k
\]
for the spectral decomposition of $A_n$. For the $(w_1,w_n)$-entry of $E_k$, note that there is a single path $P$ of length $2n-2$ from $w_1$ to $w_n$. The product of the weights of $P$ equals
\[
\omega(P) = \left(\prod_{j = 1}^{n-1} \sqrt{j(n-j)}\right)\left(\prod_{j = 1}^{n} \sqrt{j(n+1-j)}\right) = ((n-1)!)^2\sqrt{n}
\]
and that $X_n \setminus P$ is an isolated vertex, so that equation \eqref{eq:idempotent_entry} yields for $k > 0$:
\[
\begin{split}
(E_{\pm k})_{w_1,w_n} &= \lim_{t \to \pm k} t(t -(\pm k))((n-1)!)^2\sqrt{n} \prod_{j=1}^n\frac{1}{(t - j)(t+j)} \\
&= \frac{1}{2}((n-1)!)^2\sqrt{n} \prod_{\substack{j=1 \\ j\neq k}}^n\frac{1}{(k-j)(k+j)}
\end{split}
\]
Note that the sign of the $(w_1,w_n)$-entry of $E_{\pm k}$ does not depend on the sign of $\pm k$. We can further rewrite this expression as follows:
\[
\begin{split}
\prod\limits_{\substack{j = 1 \\ j \neq k}}^n  (k-j)(k+j) &= \prod_{j < k} (k-j) \cdot \prod_{j < k} (k+j)\cdot \prod_{j > k} (k-j) \cdot \prod_{j > k} (k+j) \\
&= (k-1)! \cdot \frac{(2k-1)!}{k!} \cdot (-1)^{n-k}(n-k)!\cdot \frac{(n+k)!}{2k} \\
&= \frac{(k-1)!}{k!}\cdot\frac{(2k-1)!}{(2k)!} \cdot (-1)^{n-k}(n-k)!(n+k)! \\
&= \frac{(-1)^{n-k}(n-k)!(n+k)!}{2k^2}.
\end{split}
\]
Then, using the identity $\binom{2n}{n+k} = \frac{(2n)!}{(n+k)!(n-k!)}$, we can write
\[
(E_{\pm k})_{w_1,w_n} = \frac{(-1)^{n-k}k^2(n-1)!^2\sqrt{n}}{(2n)!}\binom{2n}{n+k}.
\]
This shows that $(E_{\pm k})_{w_1,w_n}$ is positive whenever $k$ has the same parity as $n$ and negative when the parities differ. Moreover, the mutual eigenvalue support $\ev(w_1,w_n)$ contains all eigenvalues of $X_n$. The value of $D$ in Theorem \ref{thm:peakST-char-time} is $1$, and if we take
\[
\theta_s = \begin{cases}n &\text{if $n$ is odd;}\\
n-1 &\text{if $n$ is even,}
\end{cases}
\]
then $\theta_s - (\pm k)$ is even precisely whenever $k$ has the same parity as $n$ and property (ii) of Theorem \ref{thm:peakST-char-time} is also satisfied. This shows that there is peak state transfer between $w_1$ and $w_n$ at all times that are odd multiples of $\pi$. 

To determine the transfer probability, we can write
\[
\begin{split}
B(A_n)_{w_1,w_n} &= 2\sum_{k=1}^n |(E_k)_{w_1,w_n}| \\
&= \frac{2(n-1)!^2\sqrt{n}}{(2n)!}\sum_{k=1}^n k^2 \binom{2n}{n+k} \\
&= \frac{2(n-1)!^2\sqrt{n}}{(2n)!} \cdot  n\cdot 4^{n-1}.
\end{split}
\]

Here, the third equality can be verified by induction using the identity
\[
\binom{2n+2}{n+k+1} = \binom{2n}{n+k-1} + 2\textbf{}\binom{2n}{n+k} + \binom{2n}{n+k+1}
\]
and reindexing two of the three resulting sums. The sequence $(B(A_n)_{w_1,w_n})_{n=1}^\infty$ is positive and decreasing, for we can write
\[
\begin{split}
B(A_{n+1})_{w_1,w_n} &= \frac{4\sqrt{n+1}}{\sqrt{n}} \cdot \frac{n(n+1)}{(2n+1)(2n+2)}\cdot B(A_n)_{w_1,w_n} \\
&= \frac{2\sqrt{n(n+1)}}{2n+1} B(A_n)_{w_1,w_n},
\end{split}
\]
the square of the factor in front of $B(A_n)_{w_1,w_n}$ being
\[
\frac{4n^2 + 4n}{4n^2 + 4n + 1} < 1.
\]
Hence the sequence converges. We can further rewrite
\[
\begin{split}
B(A_n)_{w_1,w_n} = \frac{2 \cdot 4^{n-1}}{\sqrt{n}(n+1)C_n},
\end{split}
\]
where
\[
C_n \coloneqq \frac{1}{n+1}\binom{2n}{n} = \frac{(2n)!}{n!(n+1)!}
\]
denotes $n$-th Catalan number. A well-known Stirling-based asymptotic for the Catalan number, found for example in \cite[Chapter II.6]{FlajoletSedgewick2009}, is the following:
\[
\frac{n^{3/2}\sqrt{\pi}}{4^n}\cdot C_n \rightarrow 1 
\]
as $n \to \infty$.
This implies that
\[
\begin{split}
\lim_{n\to \infty} B(A_n)_{w_1,w_n} &= \lim_{n\to \infty} \frac{2 \cdot 4^{n-1}}{\sqrt{n}(n+1)C_n} \cdot \frac{n^{3/2}\sqrt{\pi}}{4^n}\cdot C_n \\
&= \lim_{n\to \infty} \frac{n\sqrt{\pi}}{2(n+1)} \\
&= \frac{\sqrt{\pi}}{2}.
\end{split}
\]
In particular, the transfer probability converges to $\pi/2$.
\end{proof}

As a consequence, the lower bound on the fidelity given by Lemma \ref{lem:fidelity} converges to 
\[
\frac{\sqrt{2}\pi^{\frac{1}{4}}}{1 + \frac{\sqrt{\pi}}{2}}
\approx 0.998179.
\]

Turning our attention to the family of unweighted graphs $Y_n$, we can use this to conclude the following about peak state transfer in this family. Let us denote its $(n^2+n)\times(n^2+n)$ adjacency matrix by $A_n'$. Here, `measuring at $W_n$' simply means that the measurement outcome lies in the set $W_n$.

\begin{corollary}
\label{thm:unweighted_pathgraphs_peakST}
The graph $Y_n$ admits peak state transfer, with respect to its adjacency matrix $A_n'$, from the single vertex $x \in W_1$ to every vertex $y \in W_n$, at time $\pi$. For all such $y$, the sequence $(B(A'_n)_{x,y})_{n=1}^{\infty}$ is decreasing and converges to $\sqrt{\pi}/(2\sqrt{n})$. Consequently, if the walk starts at $x$, the probability to measure at $W_n$ at time $\pi$ converges to (and is at least) $\pi/4 \approx 0.785398$.
\end{corollary}
\begin{proof}
We can write $A_n' = NA_nN^\top$, where $N$ is the block-diagonal $(n^2 + n)\times (2n)$-matrix given by
\[
N = \begin{bmatrix}
\frac{1}{\sqrt{n}}\ones_n  & & & & & \\
& \frac{1}{\sqrt{1}}\ones_1 & & & & \\
& & \frac{1}{\sqrt{n-1}}\ones_{n-1} & & & \\
& & & \frac{1}{\sqrt{2}}\ones_2 & & \\
& & & & \ddots & \\
& & & & & \frac{1}{\sqrt{n}}\ones_n
\end{bmatrix}
\]
and where $\ones_k$ is the all-ones vector of length $k$. If the spectral decomposition of $A_n$ is given by
\[
A_n = \sum_{k \in \{\pm 1,\ldots,\pm n\}} kE_k
\]
as in the proof of Theorem \ref{thm:weighted_path_peakST}, then it is easy to verify that the spectral decomposition of $A_n'$ is given by
\[
A_n' = 0 \cdot (I - NN^\top) + \sum_{k \in \{\pm 1,\ldots,\pm n\}} k NE_kN^\top.
\]
That is, the eigenvalues of $A_n$ are also eigenvalues of $A_n'$ and the latter matrix has $0$ as an additional eigenvalue. Let $y \in W_n$; since $I - NN^\top$ is block-diagonal, $0$ is not in the mutual eigenvalue support of $x$ and $y$. Moreover, for every eigenvalue $k$, the $(x,y)$-entry of $NE_kN^\top$ simply equals $(E_k)_{w_1,w_n}/\sqrt{n}$, so the sign does not change and the positive/negative mutual eigenvalue supports of all eigenvalues are the same for $A_n$ and $A_n'$. This shows that there is peak state transfer from $x$ to $y$ for all $y \in W_n$ at time $\pi$ by Theorem \ref{thm:weighted_path_peakST}. The theorem also implies that the sequence
\[
B(A_n')_{x,y} = \frac{1}{\sqrt{n}}B(A_n)_{w_1,w_n}, \quad n=1,2,\ldots
\]
is decreasing and converges to $\sqrt{\pi}/(2\sqrt{n})$, and thus a measurement lands in $W_n$ with probability at least $|W_n| = n$ times the square of this value, i.e.\ $\pi/4$.
\end{proof}

The unweighted graphs $Y_n$ can be used to transfer qubit states over the same distance and with the same fidelity as the weighted paths $X_n$, but it takes an additional unitary operation on the qubits at the vertices in $W_n$. The procedure is as follows. At time $0$ we set the qubit at the single vertex $x \in W_1$ to the qubit state $\ket{\psi}$ that we want to transfer to a qubit in $W_n$, and then at the time of peak state transfer, we measure all qubits that correspond to the vertices that are not in $W_n$. The probability of all of these qubits collapsing to the ground state $\ket{0}$ is equal to
\[
p \coloneqq |\alpha|^2 + |\beta|^2 \sum_{y \in W_n}^n B_{x,y}(A_n')^2 = |\alpha|^2 + |\beta|^2 B_{w_1,w_n}(A_n)^2,
\]
which is the same as the probability for transferring the state from $w_1$ to $w_n$ in $X_n$. The state on the remaining vertices in $W_n$ then becomes
\[
\frac{1}{\sqrt{p}}\left(\alpha\ket{0^n} + \beta\sum_{y \in W_n} B_{x,y}(A_n') \ket{\bar{y}}\right) = \frac{1}{\sqrt{p}}\left(\alpha\ket{0^n} + \beta B_{w_1,w_n}(A_n) \cdot \frac{1}{\sqrt{n}}\sum_{y \in W_n}^n \ket{\bar{y}}\right) .
\]
Applying a suitable unitary operation on the qubits corresponding to $W_n$, which sends $\ket{0^n} \mapsto \ket{0^n}$ and, for a specific vertex $y_0 \in W_n$,
\[
\frac{1}{\sqrt{n}}\sum_{j=1}^n \ket{\bar{y}} \mapsto \ket{\bar{y}_0},
\]
then results in the state
\[
\frac{1}{\sqrt{p}} (\alpha\ket{0^n} + \beta B_{w_1,w_n}(A_n) \ket{\bar{y}_0}).
\]
The qubit state $\ket{\phi} = \alpha\ket{0} + \beta B_{w_1,w_n}(A_n) \ket{1}$ at $y_0$ is then equal to the state that we would obtain at $w_n$, when transferring $\ket{\psi}$ from $w_1$ to $w_n$ in the weighted path $X_n$.

The role of this additional unitary is to localize the state that has already been transferred to the target set $W_n$. After the successful measurement, the vacuum component is represented by $\ket{0^n}$, while the single-excitation component is uniformly distributed over the vertices of $W_n$. The unitary acts only on the qubits corresponding to $W_n$, fixes $\ket{0^n}$, and maps the uniform single-excitation state to the state with its excitation at a chosen output vertex $y_0$. The resulting qubit state at $y_0$ is precisely the state obtained at $w_n$ in the corresponding weighted path $X_n$. Thus, the additional operation serves to concentrate the distributed output at a single designated vertex; its implementation cost depends on the available operations and connectivity within the target set.

\subsection{Peak state transfer in small graphs}

We numerically searched small connected graphs for the occurrence of peak state transfer with respect to both the adjacency and Laplacian matrix, using SageMath \cite{sagemath} with NumPy \cite{HarrisWillmanWalt2020}. The results are summarized in Table \ref{tab:peakST_table}, including a comparison with the occurrence of perfect state transfer. Table \ref{tab:peakST_table_trees} contains similar information, but for the class of trees specifically. The code that we used can be described as follows:

\begin{enumerate}
\item With Sagemath, we can generate all connected graphs of a given order, or all trees of a given order, and determine their adjacency and Laplacian matrices.
\item For any graph in the desired family, we first run Numpy's built-in \texttt{numpy.linalg.eigh} function to (numerically) obtain, of the appropriate matrix, a list of eigenvalues with multiplicity and a list of corresponding eigenvectors of the appropriate matrix.
\item We then group the
eigenvalues into clusters of values lying within $10^{-6}$ of each other and sum the
rank-$1$ eigenprojectors in each cluster to obtain the spectral idempotents. This groups the eigenvalue corresponding to the same eigenvalue in exact arithmetic, which is possible since the matrix has integer entries and thus its
characteristic polynomial has integer coefficients. We then sum the corresponding rank-$1$ eigenprojectors in order to determine the spectral idempotents. 
\item We then take $\theta_s$ to be either the largest eigenvalue in the adjacency case, or the smallest eigenvalue (i.e.\ zero) in the Laplacian case.
\item If peak state transfer occurs, then for the appropriate eigenvalues $\theta_r \neq \theta_s$ there must exist $a$, $b_r$ and $D$ as specified in Theorem 5.4(i). In particular, $(\theta_s - \theta_r)^2$ should be an integer.
\item If this number is not an integer, $\theta_r$ cannot be in the support $\ev(u,v)$, so if $E_r(u,v) \neq 0$, we can remove $(u,v)$ from the set of candidate pairs between which peak state transfer might occur. (In the numerical setting, we test whether $|E_r(u,v)|$ is greater than some small $\varepsilon$.) It is very much possible that no candidates $(u,v)$ remain, in which case we move on to the next graph.
 If this number is not an integer, then condition (i) fails for any pair $(u,v)$ with
$\theta_r \in \ev(u,v)$. Hence, if $E_r(u,v) \neq 0$, we can remove $(u,v)$ from the set of
candidate pairs between which peak state transfer might occur. (In the numerical setting,
we test whether $|E_r(u,v)|$ is greater than some small $\varepsilon$.) It is very much
possible that no candidates $(u,v)$ remain, in which case we move on to the next graph.
\item If $(\theta_s - \theta_r)^2$ is an integer, we can determine $D$ by taking the product of all primes that appear with an odd power in the prime factorization of this integer. That is, $D$ is the squarefree part of the integer. We also determine $c_r = \sqrt{(\theta_s - \theta_r)^2/D}$.
\item Once $(\theta_s - \theta_r)^2$ is an integer for every $\theta_r \in \ev(u,v)$,
condition (i) of Theorem \ref{thm:peakST-char-time} is automatically satisfied, so that only
condition (ii) remains to be checked. Indeed, $c_r$ is a nonnegative rational number whose
square is the integer $(\theta_s-\theta_r)^2/D$, hence $c_r$ is an integer and
$\theta_s - \theta_r = \pm c_r \sqrt{D}$. In particular every $\theta_r \in \ev(u,v)$ lies in
$\rats(\sqrt{D})$, and since the matrix has integer entries, $\theta_r$ is an algebraic
integer and therefore a quadratic integer in $\rats(\sqrt{D})$. As $D$ is squarefree, every
such element can be written as $\tfrac{1}{2}(a_r + b_r\sqrt{D})$ for integers $a_r \equiv b_r
\pmod 2$, and $\theta_s - \theta_r \in \ints\sqrt{D}$ forces $a_r = a_s$ for all $r$. This
gives the common integer $a$ required by condition (i).
\item Since we compute a single value of $D$ for each graph rather than one for each
candidate pair, we verify that $D$ is also the squarefree part of $(\theta_s -
\theta_{r'})^2$ for every $\theta_{r'}$ lying in the mutual eigenvalue support of some
candidate pair. While this is sufficient, this is not necessarily the case: it is possible that some $\theta_r \in \ev(u,v)
\setminus \ev(u',v')$ might yield the squarefree part $D$, while some $\theta_{r'} \in
\ev(u',v') \setminus \ev(u,v)$ yields an integer $(\theta_s - \theta_{r'})^2$ whose
squarefree part is $D' \neq D$; the two pairs would then have to be treated separately,
each with its own value of $D$. This did not occur for any graph in our search. 
\item Finally, if peak state transfer occurs and $\sum_r |E_r(u,v)| = 1$, then we may also count it as an example of perfect state transfer.
\end{enumerate}

The code itself is available on arXiv as an ancillary file. 

\begin{table}[htbp]
    \centering
    \begin{tabular}{c|c|c|c|c|c}
        $n$ & Connected Graphs & PST & Peak state transfer & PST & Peak state transfer \\
        & on $n$ vertices & w.r.t. $A$ & w.r.t. $A$ & w.r.t. $L$ & w.r.t. $L$ \\
        \hline
        2 & 1 & 1 & 1 & 1 & 1 \\
        3 & 2 & 1 & 2 & 0 & 2 \\
        4 & 6 & 1 & 4 & 2 & 5 \\
        5 & 21 & 1 & 9 & 0 & 13 \\
        6 & 112 & 1 & 20 & 0 & 50 \\
        7 & 853 & 1 & 42 & 0 & 191 \\
        8 & 11117 & 5 & 98 & 198 & 1332 \\
        9 & 261080 & 3 & 221 & 0 & 15055
    \end{tabular}
    \caption{Connected graphs on $n$ vertices; those which admit PST between some pair of vertices and those which admit peak state transfer between some pair of vertices. }
    \label{tab:peakST_table}
\end{table}

\begin{table}[htbp]
    \centering
    \begin{tabular}{c|c|c|c|c|c}
        $n$ & Trees & PST & Peak state transfer & PST & Peak state transfer \\
        & on $n$ vertices & w.r.t.\ $A$ & w.r.t.\ $A$ & w.r.t.\ $L$ & w.r.t.\ $L$ \\
        \hline
        2 & 1 & 1 & 1 & 1 & 1 \\
        3 & 1 & 1 & 1 & 0 & 1 \\
        4 & 2 & 0 & 1 & 0 & 1 \\
        5 & 3 & 0 & 2 & 0 & 1 \\
        6 & 6 & 0 & 2 & 0 & 1 \\
        7 & 11 & 0 & 3 & 0 & 1 \\
        8 & 23 & 0 & 2 & 0 & 1 \\
        9 & 47 & 0 & 6 & 0 & 1 \\
        10 & 106 & 0 & 2 & 0 & 2 \\
        11 & 235 & 0 & 7 & 0 & 1 \\
        12 & 551 & 0 & 3 & 0 & 1 \\
        13 & 1301 & 0 & 14 & 0 & 1 \\
        14 & 3159 & 0 & 8 & 0 & 1 \\
        15 & 7741 & 0 & 18 & 0 & 1 \\
        16 & 19320 & 0 & 16 & 0 & 1 \\
        17 & 48629 & 0 & 45 & 0 & 1 \\
        18 & 123867 & 0 & 18 & 0 & 1 \\
    \end{tabular}
    \caption{Trees on $n$ vertices; those which admit PST between some pair of vertices and those which admit peak state transfer between some pair of vertices. }
    \label{tab:peakST_table_trees}
\end{table}

\section{Conclusions and open problems}\label{sec:conclusion}

One direction for future work is the classification of peak state transfer in trees, with respect to both the adjacency matrix and the Laplacian matrix. While perfect state transfer  in trees is extremely rare — for either $A$ or $L$ it only occurs in the the paths graphs on $2$ and $3$ vertices, see \cite{CoutinhoLiu2,CouJulSpi2024} — the phenomenon of peak state transfer appears to be more widespread.

For adjacency-based dynamics, our computations of the previous section suggests that there may be infinitely many trees that support peak state transfer. In contrast, for Laplacian dynamics, the known examples are much more constrained: star graphs (complete bipartite $K_{1,n}$) always exhibit Laplacian peak state transfer, and among trees on 10 vertices, only one other example was found to admit peak state transfer. See Figure \ref{fig:trees-peak-ex}.

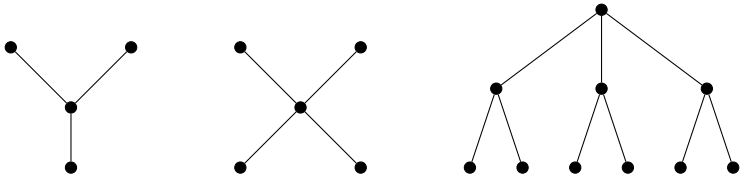
\begin{figure}[htbp]
\centering
\begin{tikzpicture}[scale=0.8, every node/.style={scale=1}]
\tikzstyle{graph node}=[circle, fill=black, draw=black, inner sep=1.5pt]

% -------- Star graph on 4 vertices --------
\node[graph node] (c1) at (0,0) {};
\node[graph node] (l1) at (-1,1) {};
\node[graph node] (l2) at (1,1) {};
\node[graph node] (l3) at (0,-1) {};
\foreach \x in {l1,l2,l3} {
    \draw (c1) -- (\x);
}

\end{tikzpicture}
\hspace{30pt}
\begin{tikzpicture}[scale=0.8, every node/.style={scale=1}]
\tikzstyle{graph node}=[circle, fill=black, draw=black, inner sep=1.5pt]

% -------- Star graph on 5 vertices --------
\node[graph node] (c2) at (0,0) {};
\node[graph node] (l1) at (-1,1) {};
\node[graph node] (l2) at (1,1) {};
\node[graph node] (l3) at (-1,-1) {};
\node[graph node] (l4) at (1,-1) {};
\foreach \x in {l1,l2,l3,l4} {
    \draw (c2) -- (\x);
}

\end{tikzpicture}
\hspace{30pt}
\begin{tikzpicture}[scale=0.7, every node/.style={scale=1}]
\tikzstyle{graph node}=[circle, fill=black, draw=black, inner sep=1.5pt]

% -------- Nearly binary tree on 10 vertices --------

% Level 0 (root)
\node[graph node] (r) at (0,0) {};

% Level 1 (3 children)
\node[graph node] (a) at (-2,-1.5) {};
\node[graph node] (b) at (0,-1.5) {};
\node[graph node] (c) at (2,-1.5) {};

% Level 2 (2 children per node)
\node[graph node] (a1) at (-2.5,-3) {};
\node[graph node] (a2) at (-1.5,-3) {};

\node[graph node] (b1) at (-0.5,-3) {};
\node[graph node] (b2) at (0.5,-3) {};

\node[graph node] (c1) at (1.5,-3) {};
\node[graph node] (c2) at (2.5,-3) {};

% Edges
\foreach \child in {a,b,c} {
    \draw (r) -- (\child);
}

\foreach \parent/\x/\y in {a/a1/a2, b/b1/b2, c/c1/c2} {
    \draw (\parent) -- (\x);
    \draw (\parent) -- (\y);
}

\end{tikzpicture}
    \caption{Three examples of trees admitting Laplacian peak state transfer found by computation; $K_{1,3}$, $K_{1,4}$, and a nearly binary tree with 10 vertices.
    \label{fig:trees-peak-ex}}
\end{figure}

A natural open problem is the following:

\begin{openprob} Determine whether infinitely many such trees, not isomorphic to the star graph, admit Laplacian peak state transfer.  \end{openprob}

Another thing: $X$ on $n$ vertices with $m$ edges and peak state transfer between $u,v$ at distance $d$. What is a lower bound on $m$? This problem is posed for perfect state transfer in \cite{CoutinhoGuo2024}. 

Another natural question concerns the relationship between graph sparsity and the ability to support peak state transfer. Suppose we have a graph  $G$ on $n$ vertices with $m$ edges and peak state transfer between $u,v$ at distance $d$. How small can $m$ be? More formally, we have the following:

\begin{openprob} If $G$ admits peak state transfer between $u,v$ at distance $d$ and has $m$ edges, find a lower bound for $m$ in terms of $d$.  \end{openprob}

In other words: what is the cheapest possible graph, in terms of the number of edges, with peak state transfer over distance $d$? A version of this problem was posed for the case of perfect state transfer in \cite{CoutinhoGuo2024}. 

Lastly, we turn our attention to strongly cospectral vertices; strong cospectrality was defined to study perfect state transfer, see for example \cite{God2010}. If $u,v$ are strongly cospectral, then the entry $B(A)_{u,v}$ of the bounding matrix for the adjacency matrix is equal to $1$. In this case, peak state transfer is equivalent to perfect state transfer. There are many graphs with strongly cospectral pairs of vertices that do not admit perfect state transfer. In this case, one can ask if there are tighter upper bounds on the maximum amount of state transfer.

\subsection*{Acknowledgements}

Gabriel Coutinho acknowledges the support from FAPEMIG and CNPq.

\subsection*{Author contributions}

All authors contributed equally to this work.

\appendix
\section{Nomenclature}\label{app:nomen}

To aid readability, we include below a table of commonly used symbols and their meanings. This allows the reader to easily recall definitions without having to search through the main text.

\begin{table}[h]
\centering
\begin{tabular}{c|l}
   $M$  & any symmetric real matrix  \\
  $M_{u,v}$ & $(u,v)$-entry of $M$ \\
   $H$ & the $2^n \times 2^n$ Hamiltonian  \\
   $G$ & graph \\ 
  $\theta_r$   &  distinct eigenvalues \\
  $\lambda,\mu$ & eigenvalues, not necessarily distinct \\ 
  $E_{r}$ & eigenprojector of $\theta_r$ \\
  $i$ & $\sqrt{-1}$\\
  $j,k,\ell, r,s,t$ & indices \\
  $n$ & number of vertices in graph \\
  $m$ & number of edges in graph \\
  $t$ & time variable \\
  $\tau$ & specific time, for example of peak state transfer \\
  $u,v,w,x,y,z$ & vertices \\
\end{tabular}
\end{table}

\section{Proof of Lemma \ref{lem:weighted_adj_eigenvalues}}
\label{app:proofs}

We will give the proof of Lemma \ref{lem:weighted_adj_eigenvalues} below. The lemma states that the eigenvalues of the $2n \times 2n$ matrix
\[
A_n =
\begin{bmatrix}
0 & \sqrt{b_1} & & \\
\sqrt{b_1} & 0 & \sqrt{b_2} & \\
& \ddots & \ddots & \ddots\\
& & \sqrt{b_{2n-2}} & 0 & \sqrt{b_{2n-1}} \\
& & & \sqrt{b_{2n-1}} & 0
\end{bmatrix},
\]
where
\[
b_{2k-1} = k(n-k+1) \quad\text{and}\quad b_{2k} = k(n-k),
\]
are precisely the integers in the set $\{\pm 1, \ldots, \pm n\}$. We will make use of Corollary 2.2 from Chapter 8 of \cite{GodsilAlgebraicCombinatorics}, which implies that the characteristic polynomial $p_{k+1}(x)$ of the $k \times k$ leading principal submatrix of $A_n$ follows the three-step recursion
\[
p_{k+1}(x) = xp_{k}(x) - b_kp_{k-1}(x).
\]
(This can be easily verified inductively by expanding the determinant of the matrix $xI - A_n$ along the last row.) If we set $p_{0}(x) = 1$ for the $0\times 0$ submatrix, the recursion is valid for all $k \geq 1$. In the proof, we will omit the variable $x$ from the polynomial notation and write $p_{k} \coloneq p_{k}(x)$.

\begin{proof}[Proof of Lemma \ref{lem:weighted_adj_eigenvalues}]
In our setting, we want to show that the characteristic polynomial of $A_n$, which has degree $2n$, is given by
\[
\prod_{k=1}^n (x^2 - k^2).
\]
By the statement above, we can express this characteristic polynomial recursively using the leading principal submatrices of $A_n$. Note that it is not the case that the top-left $(2k \times 2k)$-submatrix of $A_n$ is the matrix $A_k$, so we cannot use the recursion directly to relate $A_n$ to $A_{n-1}$.

Since the expression for $b_j$ depends on the parity of $j$, we will define two sets of polynomials: we set $p_{0}^{(n)}= 1$ and $q_{0}^{(n)} = x$ and then recursively
\[
\begin{split}
p_{k}^{(n)} &= xq_{k-1}^{(n)} - k(n-k+1)p_{k-1}^{(n)} \quad \text{for } 1 \leq k \leq n \quad \text{and}\\
q_{k}^{(n)} &= xp_{k}^{(n)} - k(n-k)q_{k-1}^{(n)}\quad \text{for } 1 \leq k \leq n-1.
\end{split}
\]
Indeed, that makes it so that $p_k^{(n)}$ is the characteristic polynomial of the leading principal $(2k \times 2k)$-submatrix of $A_n$, and similarly for $q_{k}^{(n)}$ and the leading principal $((2k+1)\times (2k+1))$-submatrix of $A_n$. Our goal is to show that for all $n \geq 0$:
\[
p_{n+1}^{(n+1)} = (x^2 - (n+1)^2) p_{n}^{(n)}.
\]
Then the statement of the lemma follows by induction.

We claim that is suffices to prove that the polynomials $p_n^{(n+1)}$ and $q_{n-1}^{(n+1)}$ corresponding to $A_{n+1}$ can be written in terms of the polynomials $p_k^{(n)}$ and $q_k^{(n)}$ corresponding to the smaller matrix $A_n$ as follows:
\begin{align}
\label{eq:claim_recurs1} p_n^{(n+1)} &= (n+1)p_{n}^{(n)} - nxq_{n-1}^{(n)}.
\quad \text{and} \\
\label{eq:claim_recurs2} q_{n-1}^{(n+1)} &= (n+1)q_{n-1}^{(n)} - nxp_{n-1}^{(n)} 
\end{align}
Assume that the claim holds. We can rewrite $p_{n+1}^{(n+1)}$, using the recursive definition, as

\begin{align}
\nonumber p_{n+1}^{(n+1)} &= xq_{n}^{(n+1)} -(n+1)p_{n}^{(n+1)} \\
\nonumber &= x(xp_{n}^{(n)} -nq_{n-1}^{(n+1)}) - (n+1)p_n^{(n+1)} \\
\label{eq:rewrite_p(n+1)(n+1)} &= (x^2 - (n+1))p_n^{(n+1)} - nxq_{n-1}^{(n+1)},
\end{align}
and then we can plug \eqref{eq:claim_recurs1} and \eqref{eq:claim_recurs2} into \eqref{eq:rewrite_p(n+1)(n+1)} to obtain

\[
\begin{split}
p_{n+1}^{(n+1)} &= (x^2 - (n+1))\left((n+1)p_n^{(n)} - nxq_{n-1}^{(n)}\right) - nx\left((n+1)q_{n-1}^{(n)} - nxp_{n-1}^{(n)}\right) \\
&= ((n+1)x^2 - (n+1)^2)p_n^{(n)} - nx^3q_{n-1}^{(n)} + n^2x^2p_{n-1}^{(n)} \\
&= ((n+1)x^2 - (n+1)^2)p_n^{(n)} - nx^2p_n^{(n)} \\
&= (x^2 - (n+1)^2)p_n^{(n)},
\end{split}
\]
as required.

To prove the claim, we will show by induction on $k$ that
\begin{align}
\label{eq:claim_recurs3} p_k^{(n+1)} &= p_{k}^{(n)} - k^2p_{k-1}^{(n)} \quad \text{and} \\
\label{eq:claim_recurs4} q_k^{(n+1)} &= q_{k}^{(n)} - k(k+1)q_{k-1}^{(n)}.
\end{align}
This holds in particular for $k = n$ in \eqref{eq:claim_recurs3} and for $k=n-1$ in \eqref{eq:claim_recurs4}. In the resulting expressions, we can then substitute
\[
np_{n-1}^{(n)} = x q_{n-1}^{(n)} - p_{n}^{(n)} \quad \text{and} \quad (n-1)q_{n-2}^{(n)} = x p_{n-1}^{(n)} - q_{n-1}^{(n)}
\]
to obtain \eqref{eq:claim_recurs1} and \eqref{eq:claim_recurs2}. We conclude the proof with this induction argument. For $k=1$, we have
\[
p_1^{(n+1)} = xq_0^{(n+1)} - (n+1)p_{0}^{(n+1)} = x^2 - n - 1 = p_1^{(n)} - 1^2 p_{0}^{(n)}
\]
and
\[
q_1^{(n+1)} = x p_{1}^{(n+1)} - nq_{0}^{(n+1)} = x\left(p_1^{(n)} - p_0^{(n)}\right) - nx = \left(xp_1^{(n)} - 1\cdot (n-1)x\right) -2x = q_1^{(n)} - 2q_0^{(n)}.
\]
Finally, we inductively find for $k \geq 2$:
\[
\begin{split}
p_k^{(n+1)} &= xq_{k-1}^{(n+1)} - k(n-k+2)p_{k-1}^{(n+1)}\\
&= x\left(q_{k-1}^{(n)} - (k-1)kq_{k-2}^{(n)}\right) \,\,-\,\, k(n-k+2)\left(p_{k-1}^{(n)} - (k-2)^2p_{k-2}^{(n)}\right) \\
&= \left(xq_{k-1}^{(n)} -k(n-k+2)p_{k-1}^{(n)}\right) \,\,-\,\, k(k-1)\left(xq_{k-2}^{(n)} - (k-1)(n-k+2) p_{k-2}^{(n)}\right) \\
&= \left(p_{k}^{(n)} - kp_{k-1}^{(n)}\right) \,\,-\,\, k(k-1)p_{k-1}^{(n)} \\
&= p_{k}^{(n)} - k^2p_{k-1}^{(n)}
\end{split}
\]
and
\[
\begin{split}
q_{k}^{(n+1)} &= xp_{k}^{(n+1)} - k(n-k+1)q_{k-1}^{(n+1)} \\
&= x\left(p_k^{(n)} - k^2p_{k-1}^{(n)}\right) \,\,-\,\, k(n-k+1)\left(q_{k-1}^{(n)} - (k-1)kq_{k-2}^{(n)}\right) \\
&= \left(x p_{k}^{(n)} - k(n-k+1)q_{k-1}^{(n)}\right) \,\,-\,\, k^2\left(xp_{k-1}^{(n)} -(k-1)(n-k+1)q_{k-2}^{(n)}\right) \\
&= \left(q_{k}^{(n)} - kq_{k-1}^{(n)}\right) \,\,-\,\, k^2 q_{k-1}^{(n)} \\
&= q_{k}^{(n)} - k(k+1)q_{k-1}^{(n)}. \qedhere
\end{split}
\]
\end{proof}

\end{document}